\documentclass[12pt, draftclsnofoot, onecolumn]{IEEEtran}

\fontfamily{cmr}\selectfont

\usepackage{wrapfig}
\usepackage{algorithm}
\usepackage{algpseudocode}
\usepackage{pifont}
\usepackage{enumitem} 
\usepackage{cite}
\usepackage{amsmath}
\usepackage{times}
\usepackage{textcomp}
\usepackage{amsfonts}
\usepackage{amssymb}
\usepackage{amsthm}
\usepackage{graphicx}
\usepackage[utf8]{inputenc}
\usepackage{authblk}
\usepackage{cancel} 
\usepackage{graphicx}%
\usepackage{float}
\usepackage{layout}
\usepackage{array}%
\usepackage{subfig}
\usepackage{comment}
\usepackage{array}%
\usepackage{color}%
\usepackage[usenames,dvipsnames]{xcolor}
\usepackage{soul}
\usepackage{footmisc}%
\theoremstyle{plain}
\usepackage{epstopdf}
\usepackage{amssymb}
\usepackage{bbm}
\usepackage{mathtools}
\usepackage{cuted}
\usepackage{lipsum, color}
\usepackage{tikz}
\usepackage{tkz-tab}
\usetikzlibrary{automata,arrows,positioning,calc}
\usepackage{caption}

\newtheorem{theorem}{Theorem}

\newtheorem{definition}{Definition}

\newtheorem{proposition}{Proposition}
\newtheorem{remark}{Remark}

\renewcommand{\vec}[1]{\mathbf{#1}}
\newcommand{\mI}{\mathbb{I}}
\newcommand{\mA}{\mathcal{A}}

\DeclareMathOperator*{\argmax}{arg\,max}
\DeclareMathOperator*{\argmin}{arg\,min}

\newcommand*\circled[1]{\tikz[baseline=(char.base)]{\node[shape=circle,draw,inner sep=0.5pt] (char) {#1};}}

\begin{document}


\title{On the Performance of Mobility-Aware \\ D2D Caching Networks}

\author{
    \IEEEauthorblockN{Sameh Hosny, Atilla Eryilmaz, Alhussein A. Abouzeid and Hesham El Gamal}}


\maketitle
%


\begin{abstract}

The increase in demand for spectrum-based services forms a bottleneck in wireless networks. Device-to-Device (D2D) caching networks tackle this problem by exploiting users behavior predictability and the possibility of sharing data between them to alleviate the network congestion. Usually, network congestion occurs at certain times of the day and in some popular locations. Consequently, the information about user demand alone is not enough. Capturing mobility statistics allows Service Providers (SPs) to enhance their caching strategies. In this work, we introduce a mobility-aware D2D caching network where an SP harnesses users demand and mobility statistics to minimize the incurred service cost through an optimal caching policy. We investigate two caching schemes: a centralized caching scheme and a decentralized caching scheme. In the centralized caching scheme, the SP makes the caching decision towards its cost minimization to increase its profit. However, the complexity of the optimal caching policy grows exponentially with the number of users. Therefore, we discuss a greedy caching algorithm which has a polynomial order complexity. We also use this greedy algorithm to establish upper and lower bounds on the proactive service gain achieved by the optimal caching policy. In the decentralized caching scheme, users take over and make their caching decisions, in a distributed fashion affected by the SP pricing policy, towards their payment minimization. We formulated the tension between the SP and users as a \emph{Stackelberg game}. Best response analysis was used to identify a \emph{subgame perfect Nash equilibrium} (SPNE) between users. The optimal solution of the proposed model was found to depend on the \emph{SP reward preference}, which affects the assigned memory in users devices. We found some regimes for the reward value where the SPNE was non-unique. A \emph{fair allocation} caching policy was adopted to choose one of these SPNEs. To understand the impact of user behavior, we investigated some special cases to explore how users mobility statistics affect their caching decision. The obtained results in this work allow us to enhance our previously studied content trading model \cite{Hosny2015Game} to form a complete vision of mobile content trading. Based on the results obtained in this work, we plan to formulate a mobility-aware content trading marketplace. We expect to achieve more gains by exploiting the users mobility statistics when they are allowed to trade their proactive downloads.

\end{abstract}

%


%
\section{Introduction}

The growth in data traffic represents a crucial problem in mobile networks. More than half a billion mobile devices were added in 2015 causing a 74$\%$ growth in global mobile data traffic. Nevertheless, an eightfold increase in this traffic is expected between 2015 and 2020. Moreover, three-fourths of the world’s mobile data traffic will be video by 2020 \cite{cisco_2015}. This increase in demand for spectrum-based services and devices has led network SPs to experience a major demand and supply mismatch during the whole day \cite{Federal2002Spectrum}. This demand disparity is ultimately tied to user behavioral pattern. However, most people follow certain daily routines and hence their behavior is highly predictable \cite{Song2010Limits},\cite{farrahi2008discovering}. Interestingly, the time-varying user activities, that are ultimately contributing to this mismatch, can be exploited to solve this demand disparity.

The concept of \emph{proactive resource allocation} for wireless networks was established to control the supplied services to best match the demand patterns \cite{Tadrous2013Proactive}. The predictability of user behavior is exploited to balance the wireless traffic over time, and significantly reduce the bandwidth required to achieve a given blocking/outage probability. \emph{Device-to-device (D2D) communication} has been proposed in \cite{yu2011resource} as a promising technology that can relief the wireless networks congestion. A pair of end-users, moving within a close proximity to each other, establish a D2D link that can be operated in the unlicensed spectrum band, such as the Industrial, Scientific, and Medical (ISM) radio bands. These D2D links when used as a traffic offloading approach introduces very little or no monetary cost for the end-users.

A tutorial overview of some recent results on base station assisted D2D wireless networks with caching for video delivery was presented in \cite{Caire2016Review}. Some competing conventional schemes and a recently developed scheme based on caching at the user devices was also introduced. Throughput-outage scaling laws of such schemes were discussed. It was shown that, in realistic conditions, the D2D caching scheme largely outperforms all other competing schemes both in terms of per-user throughput and in terms of outage probability. A D2D caching network under arbitrary demand was considered in \cite{Caire2016CachingD2D}. It was shown that if each node in the network can reach in a single hop all other nodes, then the proposed scheme achieves almost the same throughput of \cite{Maddah2014Fundamental}. Moreover, if concurrent short range transmissions can co-exist in a spatial reuse scheme, then the throughput has the same scaling law of the reuse-only case \cite{Caire2013Optimal, Caire2015Throughput} or the coded-only case \cite{Maddah2014Fundamental}. 
Although previous models utilized the D2D communication to alleviate network congestion, they considered a grid network formed by a set of nodes placed on a regular grid on the unit square and user's mobility was not captured in this work.

The authors in \cite{DavidTse2002Mobility} considered the model of \cite{Gupta2000Capacity} and showed that the per-user throughput can increase dramatically when nodes are mobile rather than fixed. This improvement was obtained under several idealistic assumptions. They assumed complete mixing of nodes trajectories in the network and random mobility pattern was not considered. They also assumed that data contents are delay tolerant and stated that their ideas were not very relevant to real-time applications. Caching data contents in users devices helps us to overcome the delay constraint. Furthermore, a practical mobility model is required to represent a more realistic behavior of the users. There are many mobility models in the literature which try to capture user behavior \cite{Camp02asurvey}. In this work, we focus on the individual user mobility based on a probabilistic random walk and defer group mobility for our future work.


We consider a D2D caching network where the SP is aware of the user demand and mobility. We consider the results presented here as a forward step towards a mobile content marketplace. The obtained results will allow us to enhance our content trading model presented in \cite{Hosny2015Game}. Our aim is to show that exploiting the information about user mobility helps the SP to optimize its caching strategy and address the network congestion problem in an intelligent manner. Moreover, users can achieve more gains when they consider their mobility statistics and the locations where they can meet other users in the network. We investigate two caching schemes: a centralized caching scheme and a decentralized caching scheme. In the centralized caching scheme, the SP makes the caching decision towards its cost minimization to increase its profit. In the decentralized caching scheme, users take over and make their caching decisions, in a distributed fashion affected by the SP pricing policy, towards their payment minimization. Our main contributions are:
\begin{enumerate}
    \item We introduce an optimal centralized caching policy that allows SP to enhance its caching decisions based on the user demand and mobility statistics.
    \item The complexity of the optimal centralized caching policy grows exponentially with the number of users. Therefore, We introduce a sub-optimal policy based on a greedy algorithm that has a polynomial order complexity.
    \item Using the proposed greedy algorithm, we establish upper and lower bounds on the gain achieved by the optimal policy for the proactive service cost.
    \item We investigated how the SP chooses an optimal reward to incentives users to participate in the proposed centralized caching policy.
    \item We extend our work by considering a decentralized caching policy. The tension between the SP and users is modeled as a Stackelberg game. Best response analysis was used to identify a \emph{subgame perfect Nash equilibrium} (SPNE) between users.
    \item The optimal solution of the proposed model was found to depend on the SP reward preference, which affects the assigned memory in users devices. We found some regimes for the reward value where the SPNE was non-unique. A \emph{fair allocation} caching policy was adopted to choose one of these SPNEs.
    \item We studied the relation between the users assigned memory and the reward they receive from the SP. This part studies the tension between the SP and users to choose an appropriate memory size for the decentralized caching policy.
    \item To understand the impact of user mobility, we considered some special cases when users have similar behavior. We studied the effect of these special cases on the centralized and decentralized caching policies.
    \item The obtained results in this work allow us to enhance our content trading model presented in \cite{Hosny2015Towards} to form a complete vision about mobile content marketplace.
\end{enumerate}

The rest of this paper is organized as follows. In Section \ref{Sec:System_Model}, we lay out the system setup and define the characteristics of its main components. We study the performance of the centralized caching scheme in Section \ref{Sec:Centralized_Caching}. In Section \ref{Sec:Decentralized_Caching}, we study the decentralized caching scheme. The paper is concluded in Section \ref{Sec:Conclusion}.


\section{System Model} \label{Sec:System_Model}

We consider a wireless network consisting of a set of $N$ \emph{users} $\mathcal{N} = \{1,2,\cdots,N\}$ and a single \emph{Service Provider} (SP) who supplies $M$ \emph{data items} $\mathcal{M} = \{1,2,\cdots,M\}$ upon demand. Each data item $m \in \mathcal{M}$ has a size $S_m>0$ which may be a movie (as in YouTube and Netflix), a sound track (as in Panadora), a social network update (as in Facebook and Twitter), a news update (as in CNN and Fox News), etc. Each user may request any of these data items in a random fashion. We consider a time-slotted system where SP divides the duration of interest (e.g. a day) into $T$ time slots. We assume that the duration of each slot is the time taken for a user to completely consume the requested data item and hence each time slot is in the order of minutes or possibly hours. At the beginning of each time slot, SP collects the demand of all users and supplies them with the requested data items.

\subsection{User Demand Model}

We assume that SP can track, learn and predict user behavior over time and hence constructs a \emph{demand profile} for every user $n$ denoted by $\vec{\Pi}_n=\left(\vec{p}_{n,t}\right)_t$. For any time slot $t$, $\vec{p}_{n,t}=\left(p_{n,t}^{m}\right)_m$ where $p_{n,t}^{m}$ is the probability that user $n$ requests item $m$ in time slot $t$. The demand of user $n$ in time slot $t$ is captured by a random variable $\mI_{n,t}^{m}$ where 
	\begin{equation*}\label{Eq:Demand_Indicator}
	\mI_{n,t}^{m}=\begin{cases}1, & \text{with probability } p_{n,t}^{m},\\
	                     0, & \text{with probability } 1-p_{n,t}^{m}.
	\end{cases}
	\end{equation*}
We assume that at any time slot $t$, $\mI_{n,t}^m$ is independent of $\mI_{n,t+1}^{m}, \forall n,m$. We aslso assume that for any $n \neq k$, $\mI_{n,t}^{m}$ is independent of $\mI_{k,t}^{m}, \forall m,t$. Furthermore, the demand profile of each user follows a \emph{cyclo-stationary} pattern that repeats itself in a period of $T$ time slots. That is, we can write $\vec{p}_{n,t}=\vec{p}_{n,t+kT}$ for any non-negative integer $k$. As an example, the $T$-slot period can be interpreted as a single day through which the activity of each user varies each hour, but occurs with the same statistics every day. SP relates these time slots with the actual day time based on users demand statistics to recognize the time slots where it experiences low demand (off-peak times) and those where a high demand occurs (peak times).

\subsection{User Mobility Model}

We assume that SP is interested in $L$ \emph{popular} locations $\mathcal{L}=\{1,2,\cdots,L\}$ like airports, schools, shopping malls, stadiums or governmental buildings where high demand can be related with mobility of users. Moreover, SP can track, learn and predict the mobility of each user over time and hence constructs a \emph{mobility profile} for every user $n$ denoted by $\vec{\Theta}_{n}=\left(\theta_{n,t}^{l}\right)_{t}^{l}$ where $\theta_{n,t}^{l}$ is the probability that user $n$ will be present at location $l$ in time slot $t$ where $\sum_{l=1}^{L} \theta_{n,t}^{l} = 1 \hspace{1mm} \forall n,t$. We represent user's mobility by a modified probabilistic version of the random walk mobility model which is based on a discrete-time Markov chain model \cite{Chiang98Mobility}. We assume that users stay in the same location within a time slot and may move to another location at the beginning of each time slot. Let $\lambda_{n,t}^{l,k}$ be the transition probability that user $n$ moves from location $l$ to location $k$ in time slot $t$ where $\sum_{k=1}^{L} \lambda_{n,t}^{l,k} = 1 \hspace{1mm} \forall n,t,l$. These transition probabilities may change from one time slot to another to capture the mobility of each user. However, the probability of being at a certain location in a time slot $t$ depends on the location in the previous time slot $t-1$ only, i.e. $\theta_{n,t}^{l} = \sum_{k=1}^{L} \theta_{n,t-1}^{k} \lambda_{n,t}^{k,l}$ where $\theta_{n,1}^{l}=\lambda_{n,1}^{l,l}$. 

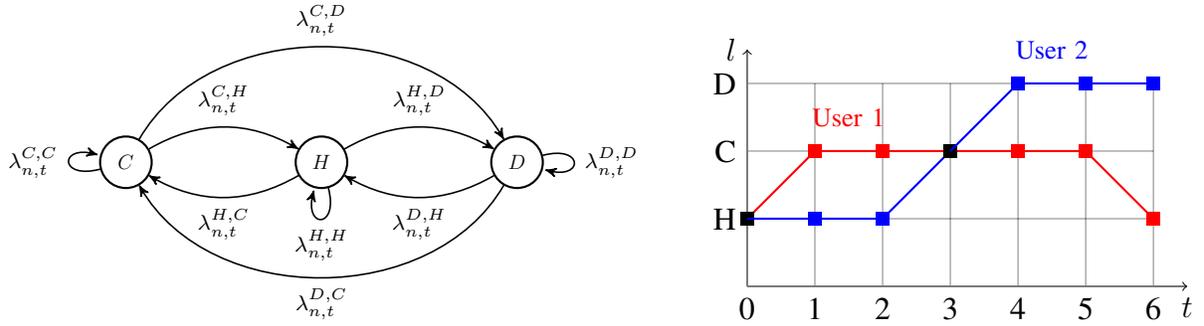
\begin{figure}[h!]
    \centering
        \subfloat[An example state transition diagram of user $n$ in time slot $t$ for $L=3$ locations, like home (H), campus (C) and downtown (D).]{
            \begin{tikzpicture}[->, >=stealth', auto, semithick, node distance=4cm]
                \tikzstyle{every state}=[fill=white,draw=black,thick,text=black,scale=0.65]
                \node[state]    (H)                 {$H$};
                \node[state]    (D)[right of= H]    {$D$};
                \node[state]    (C)[left of= H]     {$C$};
                \path
                (H) edge[loop below,below]      node{\scriptsize{$\lambda_{n,t}^{H,H}$}}    (H)
                    edge[bend left,below]       node{\scriptsize{$\lambda_{n,t}^{H,C}$}}    (C)
                    edge[bend left,above]       node{\scriptsize{$\lambda_{n,t}^{H,D}$}}    (D)
                (C) edge[loop left,left]        node{\scriptsize{$\lambda_{n,t}^{C,C}$}}    (C)
                    edge[bend left,above]       node{\scriptsize{$\lambda_{n,t}^{C,H}$}}    (H)
                    edge[bend left=60,above]    node{\scriptsize{$\lambda_{n,t}^{C,D}$}}    (D)
                (D) edge[loop right,right]      node{\scriptsize{$\lambda_{n,t}^{D,D}$}}    (D)
                    edge[bend left,below]       node{\scriptsize{$\lambda_{n,t}^{D,H}$}}    (H)
                    edge[bend left=60,below]    node{\scriptsize{$\lambda_{n,t}^{D,C}$}}    (C);
            \end{tikzpicture} } \quad     
        \subfloat[An example trajectory of two users for $L=3$. They meet at home (H) at $t=0$ and on campus (C) at $t=3$ and exchange data.]{
            \begin{tikzpicture}[scale=0.90]
                \draw[->] (0,0) -- (6.5,0) node[anchor=north] {$t$};
                \draw	(0,0) node[anchor=north] {0}
                		(1,0) node[anchor=north] {1}
                		(2,0) node[anchor=north] {2}
                		(3,0) node[anchor=north] {3}
                		(4,0) node[anchor=north] {4}
                		(6,0) node[anchor=north] {6}
                		(5,0) node[anchor=north] {5};
                \draw[->] (0,0) -- (0,3.5) node[anchor=east] {$l$};
                \draw	(0,0) node[anchor=east] {}
                		(0,1) node[anchor=east] {H}
                		(0,2) node[anchor=east] {C}
                		(0,3) node[anchor=east] {D};            
                
                \draw[step=1cm,gray,very thin] (0,0) grid (6,3);            
    
                \draw[red]	(1.5,2.5) node{\small{User 1}};
                \draw[blue] (4.5,3.5) node{\small{User 2}};
    		            
                \fill[black!100!white] (0.1,1.1) rectangle (-0.1,0.9);
                \draw[red,thick] (0,1) -- (1,2); ;
                \fill[red!100!white] (1.1,2.1) rectangle (0.9,1.9);
                \draw[red,thick] (1,2) -- (2,2);
                \fill[red!100!white] (2.1,2.1) rectangle (1.9,1.9);
                \draw[red,thick] (2,2) -- (3,2);
                \fill[black!100!white] (3.1,2.1) rectangle (2.9,1.9);
                \draw[red,thick] (3,2) -- (4,2);
                \fill[red!100!white] (4.1,2.1) rectangle (3.9,1.9);
                \draw[red,thick] (4,2) -- (5,2);
                \fill[red!100!white] (5.1,2.1) rectangle (4.9,1.9);
                \draw[red,thick] (5,2) -- (6,1);
                \fill[red!100!white] (6.1,1.1) rectangle (5.9,0.9);
                
                \fill[black!100!white] (0.1,1.1) rectangle (-0.1,0.9);
                \draw[blue,thick] (0,1) -- (1,1); ;
                \fill[blue!100!white] (1.1,1.1) rectangle (0.9,0.9);
                \draw[blue,thick] (1,1) -- (2,1);
                \fill[blue!100!white] (2.1,1.1) rectangle (1.9,0.9);
                \draw[blue,thick] (2,1) -- (3,2);
                \fill[black!100!white] (3.1,2.1) rectangle (2.9,1.9);
                \draw[blue,thick] (3,2) -- (4,3);
                \fill[blue!100!white] (4.1,3.1) rectangle (3.9,2.9);
                \draw[blue,thick] (4,3) -- (5,3);
                \fill[blue!100!white] (5.1,3.1) rectangle (4.9,2.9);
                \draw[blue,thick] (5,3) -- (6,3);
                \fill[blue!100!white] (6.1,3.1) rectangle (5.9,2.9);               
             
            \end{tikzpicture}}\\
    \caption{An example for user mobility model}
    \label{Fig:User_Mobility}
\end{figure}    
    
Figure \ref{Fig:User_Mobility} (a) shows the state transition diagram of user $n$ in time slot $t$ for $L=3$ locations, like home (H), campus (C) and downtown (D). We assume that each user randomly takes a trajectory everyday starting from one location and moving to other locations. However, we assume that the mobility profile of each user follows a \emph{cyclo-stationary} pattern that repeats itself in a period of $T$ time slots. For example, everyday user starts from his home, visits some frequent locations throughout the day and then returns back home at the end of the day as shown in Figure \ref{Fig:User_Mobility} (b). SP exploits this mobility to enhance its caching strategy and hence achieves more gain by reducing the incurred service cost.

\subsection{Proactive Service Scheme}

SP tries to smooth out the network load by caching some of these data items at the network edge and exploits users mobility statistics to enhance its caching decision. We assume that one-hop device-to-device (D2D) communication is allowed and can be used to transfer data items between users. A fixed data rate link between all users is assumed. We also consider a non-fading channel between all users where an appropriate network protocol is applied to avoid multiple access interference. In the small timescale, data transmission follows an orthogonal multiple access scheme, hence inter-node interference effect is ignored in our large timescale model. For example, at any location $l$, SP predicts that a certain item $m$ will experience a high request in time slot $t$. It also predicts which users will be possibly present at that location in this time slot. This data item can be cached at these users and they can transfer it to other users in their vicinity. Therefore, some of the network load will be shifted to the D2D communication which alleviates the network congestion and yields a reduction in the incurred service cost.

Users occupy part of their device memory for caching these data items and consume some of their batteries to transfer it through the D2D communication. We capture the cost of caching each byte by a parameter $r>0$. This parameter can be viewed as a rent cost for caching this data. We can also view it as a reward that incentives users to participate in this model and save some of their payments by getting it as a discount in their monthly bills. For simplicity, we assume that users always have enough battery level to transfer cached data items to other users in the network and that they always allow SP to cache data in their devices. This reward promotes users to raise their memory size to be able to cache more data. When $N$ is sufficiently large, we can assume that each user has enough memory space to cache assigned data items since SP distributes cached items over all available users.


\section{Centralized Caching Scheme} \label{Sec:Centralized_Caching}

In the centralized caching mode, SP makes the caching decision to push some data items in users devices. SP leverages the information about users demand and mobility statistics to make these decisions towards its cost minimization. Users get reward by participating in this model and find their request either in their local cache or at other users in the same vicinity. Therefore, users can also save some of their payments. To evaluate the system performance, we compare the incurred service cost of the proposed model with the cost of the flat pricing scenario. The definition of the cost function is first defined and then the problem is stated. We introduce an optimal centralized caching policy and resolve its complexity issue through a suboptimal caching policy. We also shed light on the impact of users mobility on the proposed caching policy.

\subsection{Problem Statement}

To supply requested data items, SP incurs a certain service cost due to the resources consumed at each time slot. We denote by $C(L_t)$ the SP cost for serving a total demand $L_t \geq 0$ in time slot $t$. We also assume that the cost function $C:\mathbb{R}^+ \rightarrow \mathbb{R}^+$ is convex and non-decreasing. We consider a \emph{reactive network} as a baseline scenario where users' requests are served upon arrival (in contrast to proactively predicting the demand requests). In this case, the time-averaged expected cost of all users is given by:
    \begin{equation}\label{Eq:Cost_Reactive}
        C^{\mathcal{R}} = \limsup_{T\to\infty} \frac{1}{T} \sum_{t=1}^{T} \mathbb{E} \left[ C \left( \sum_{n=1}^{N} \sum_{m=1}^{M} S_m p_{n,t}^{m} \right) \right].
    \end{equation}
where, the superscript $\mathcal{R}$ indicates reactive operation. In the proposed model, we assume that SP is aware of the demand and mobility profiles of all users over $T$ time slots. SP caches an amount $x_{n}^{m}$ of data item $m$ at user $n$ for a future possible request. Each user $n$ transfers this data to other users through the D2D communication in any time slot $t$ when it is requested. For simplicity, we assume that sharing the cached data in users devices happen for free. In particular, users are not announcing any selling prices and they don't pay for getting their request from other users. Therefore, when a user requests a certain content, his request will be first served from the cached data in the other users devices who are located around him. If this data content was not cached, the request will be served through the network resources. SP replaces the data cached in users devices when it is expired at the end of the day (i.e. at the end of time slot $T$). In particular, SP caches data at the beginning of the day and lets users share it throughout the rest of the day. The cached amount of data item $m$ at each user cannot exceed its size, i.e. 
    \begin{equation}\label{Eq:Const}
        0 \leq x_{n}^{m} \leq S_m, \hspace{1mm} \forall n,m
    \end{equation}
Hence, under this proactive model, the total network load in time slot $t$ is given by:
    \begin{equation}\label{Eq:Load_Proactive}
        \begin{aligned}
            L_t^{\mathcal{P}} & = \sum_{m=1}^{M} \sum_{n=2}^{N-1} \sum_{a_n \in \mA_n} \biggl(S_m - \sum_{k \in a_n} x_k^m \biggr)^+ \sum_{k \in a_n} p_{k,t}^{m} \underbrace{\sum_{l=1}^{L} \prod_{k \in a_n} \theta_{k,t}^{l} \prod_{j \notin a_n} \Bigl(1-\theta_{j,t}^{l}\Bigr)}_\text{some users are together} \\
            & + \sum_{m=1}^{M} \biggl(S_m - \sum_{n=1}^{N} x_n^m \biggr)^+ \sum_{n=1}^N p_{n,t}^{m} \underbrace{\sum_{l=1}^{L} \prod_{n=1}^N \theta_{n,t}^{l}}_\text{all users are together} \\
            & + \sum_{m=1}^{M} \sum_{n=1}^{N} \biggl( S_m - x_n^m \biggr) p_{n,t}^{m} \underbrace{\biggl( 1 - \sum_{l=1}^{L} \sum_{k=2}^{N} \sum_{a_k \in \mA_k} \prod_{j \in a_k} \theta_{j,t}^{l} \prod_{i \notin a_k} \Bigl(1-\theta_{i,t}^{l}\Bigr) \biggr)}_\text{every user is alone}  
        \end{aligned}
    \end{equation}
where the superscript $\mathcal{P}$ indicates \emph{proactive} operation and $A_n$ is the set of all possible combinations of $n$ indices. i.e. 
    \begin{equation*}
        \mA_n = \Bigl\{a_n := \bigl(k_1,\cdots,k_n\bigr),k_j \in \bigl\{1,2,\cdots,N\bigr\} \hspace{1mm} \forall j \Bigr\}
    \end{equation*}
where $|\mA_n| = \binom{N}{n}$. We assume that users share the cached data items when they meet each other and get the remaining portion from SP. The total network load in (\ref{Eq:Load_Proactive}) captures all cases when some of the users meet each other, all users meet together or each user is moving alone. Consequently, the corresponding time-averaged expected cost under the proactive operation is given by:
    \begin{equation}\label{Eq:Cost_Proactive}
        C^{\mathcal{P}} = \limsup_{T\to\infty} \frac{1}{T} \sum_{t=1}^{T} \mathbb{E} \biggl[ C \Bigl( L_{t}^{\mathcal{P}} \Bigr) \biggr] + r \sum_{n=1}^{N} \sum_{m=1}^{M} x_{n}^{m} 
    \end{equation}
which captures SP's cost for serving a proactive demand $L_t^\mathcal{P}$ and the corresponding cost for caching process. Note that instead of having a cost factor ($\beta$), as in the previous chapters, we are modeling the SP cost in terms of serving the peak load for a normalized cost factor and caching some data items for a reward factor ($r$).


The SP gain is the difference between the \emph{reactive} cost and the \emph{proactive} cost under the proposed model which can be denoted by $\bigtriangleup C = C^\mathcal{R}-C^\mathcal{P}$. Users save some of their payments by finding the requested data items in their local cache or in the cache of their neighbors. The SP objective is to achieve a positive gain (i.e. $\bigtriangleup C >0$) by finding an optimal caching policy $\{x_{n}^{m*}\}_{n,m}$, which minimizes the time-averaged expected cost, while serving the requested data items on time to all users. The problem is defined as:
    \begin{equation}\label{Eq:SP_Optimization_Problem}
        \begin{aligned}
            & \min \hspace{5mm} C^{\mathcal{P}}\\
            & \text{s.t.} \hspace{5mm} (\ref{Eq:Const}).
        \end{aligned}
    \end{equation}

The optimization problem in (\ref{Eq:SP_Optimization_Problem}) depends mainly on the cost function $C$ which may be linear, quadratic or a polynomial of higher order. The exact solution of (\ref{Eq:SP_Optimization_Problem}) for non-linear cost functions can be obtained using convex optimization techniques. However, this case does not provide clear insights on the effect of user's mobility. Nevertheless, finding an optimal caching policy will be non-tractable. Instead, we focus here on a linear cost function to reveal some insights and to find an optimal caching policy, which allows SP to achieve a minimum service cost. The complexity of this optimal policy grows exponentially with the number of users $N$. We overcome this point by introducing a suboptimal policy based on a greedy algorithm which has a polynomial-order complexity. We use the sub-optimal policy to find upper and lower bounds for the optimal policy.

\subsection{Optimal Centralized Caching Policy Analysis}\label{Sec:Centralized_Optimal}

In this section we introduce an optimal caching policy which achieves a minimum service cost for the proposed model. For a linear cost function, considering all possible cases of the $\bigl( . \bigr)^+$ terms in (\ref{Eq:Load_Proactive}), we wind up with a set of linear programs and the optimal solution is obtained from the one which leads to a minimum cost. We start by considering two simple cases for $N=2,3$. We use these simple cases to generalize the optimal solution of this centralized caching policy.

\subsubsection{\textbf{\underline{Case Study ($N=2$)}}}

For simplicity, we start by the case when $T=1$ and then extend it for any value of $T$. In this case, the suffix $t$ can be dropped and the expected load (\ref{Eq:Load_Proactive}) will be:
    \begin{equation} \label{Eq:Load_N=2_T=1}
        \begin{aligned}
            L^{\mathcal{P}} &= \sum_{m=1}^{M} \biggl(S_m - \Bigl(x_1^m+x_2^m\Bigr) \biggr)^+ \Bigl(p_1^m+p_2^m \Bigr) \sum_{l=1}^{L} \theta_1^l \theta_2^l\\
                    &+ \sum_{m=1}^{M} \biggl(1-\sum_{l=1}^{L} \theta_1^l \theta_2^l \biggr) \sum\limits_{n=1}^{2 } \Bigl(S_m - x_n^m\Bigr) p_n^m 
        \end{aligned}
    \end{equation}

And the optimization problem will be:
    \begin{equation}\label{Eq:SP_Optimization_Problem_N=2_T=1}
        \begin{aligned}
            & \min \hspace{5mm} L^{\mathcal{P}} + r \sum_{m=1}^{M} \Bigl( x_1^m+x_2^m\Bigr)\\
            & \text{s.t.} \hspace{5mm} 0 \leq x_n^m \leq S_m, \hspace{5mm} \forall n,m
        \end{aligned}
    \end{equation}
The problem decomposes to $M$ sub-problems and we have two sub-cases: either $x_1^m+x_2^m<S_m$, which leads to a linear program (LP), where:
    \begin{equation} \label{Eq:Load_N=2_T=1_Case1}
        \begin{aligned}
            L^{\mathcal{P}} = \underbrace{\sum_{m=1}^{M} \sum_{n=1}^{2} S_m p_n^m}_\text{reactive load} - \underbrace{\sum_{m=1}^{M} \sum_{n=1}^{2} x_n^m p_n^m}_\text{caching gain} - \underbrace{\sum_{m=1}^{M} \Bigl( x_1^m p_2^m + x_2^m p_1^m \Bigr) \sum_{l=1}^{L} \theta_1^l \theta_2^l}_\text{sharing gain}
        \end{aligned}
    \end{equation}
or $x_1^m+x_2^m \geq S_m$, which leads to another LP, where:
    \begin{equation} \label{Eq:Load_N=2_T=1_Case2}
        \begin{aligned}
            L^{\mathcal{P}} = \underbrace{\sum_{m=1}^{M} \sum_{n=1}^{2} S_m p_n^m}_\text{reactive load} - \underbrace{\sum_{m=1}^{M} \sum_{n=1}^{2} x_n^m p_n^m}_\text{caching gain} - \underbrace{\sum_{m=1}^{M} \sum_{n=1}^{2} \Bigl(S_m - x_n^m \Bigr) p_n^m \sum_{l=1}^{L} \theta_1^l \theta_2^l}_\text{sharing gain}
        \end{aligned}
    \end{equation}
Note that the first term in (\ref{Eq:Load_N=2_T=1_Case1}) and (\ref{Eq:Load_N=2_T=1_Case2}) represents the \emph{reactive load} of the network, the second term represents the \emph{caching gain} achieved by caching $x_1^m$ and $x_2^m$ at users 1 and 2 respectively, while the last term represents the \emph{sharing gain} attained when each user transfers his proactive download to the other. The feasibility regions of these LPs are shown in Figure \ref{Fig:Feasibility_Regions_N=2}.

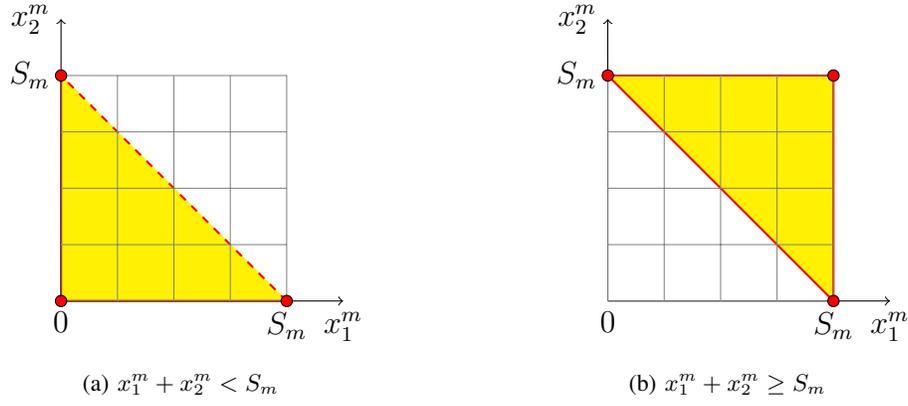
\begin{figure}[h!]
    \centering
    \subfloat[$x_1^m+x_2^m < S_m$]{
        \centering
        \begin{tikzpicture}[scale=0.75]
        
            \draw[fill=yellow,yellow]  (0,0) --  (4,0) --  (0,4) -- cycle;
            \draw[red,thick,dashed] (4,0) -- (0,4);
            \draw[red,thick] (0,0) -- (4,0);
            \draw[red,thick] (0,0) -- (0,4);
    
            \draw[->] (0,0) -- (5,0) node[anchor=north] {$x_1^m$};
            \draw	(0,0) node[anchor=north] {$0$}
                	(4,0) node[anchor=north] {$S_m$};
            \draw[->] (0,0) -- (0,5) node[anchor=east] {$x_2^m$};
            \draw	(0,0) node[anchor=east] {}
                	(0,4) node[anchor=east] {$S_m$};            
                
            \draw[step=1cm,gray,very thin] (0,0) grid (4,4);            
            \draw [fill=red] (0,0) circle (0.1cm);
            \draw [fill=red] (4,0) circle (0.1cm);
            \draw [fill=red] (0,4) circle (0.1cm);
                
        \end{tikzpicture}}
    \hspace{2cm}
    \subfloat[$x_1^m+x_2^m \geq S_m$]{
        \centering
        \begin{tikzpicture}[scale=0.75]
    
            \draw[fill=yellow,yellow]  (4,0) --  (4,4) --  (0,4) -- cycle;
            \draw[red,thick] (4,0) -- (0,4);
            \draw[red,thick] (0,4) -- (4,4);
            \draw[red,thick] (4,0) -- (4,4);
                
            \draw[->] (0,0) -- (5,0) node[anchor=north] {$x_1^m$};
            \draw	(0,0) node[anchor=north] {$0$}
                	(4,0) node[anchor=north] {$S_m$};
            \draw[->] (0,0) -- (0,5) node[anchor=east] {$x_2^m$};
            \draw	(0,0) node[anchor=east] {}
                	(0,4) node[anchor=east] {$S_m$};            
                
            \draw[step=1cm,gray,very thin] (0,0) grid (4,4);            
            \draw [fill=red] (4,4) circle (0.1cm);
            \draw [fill=red] (4,0) circle (0.1cm);
            \draw [fill=red] (0,4) circle (0.1cm);
                
        \end{tikzpicture}}
    \caption{Feasibility regions for $N=2$}
    \label{Fig:Feasibility_Regions_N=2}
\end{figure}

The optimal solution of each LP is at one of its extreme points in the corresponding feasibility region. We have 4 extreme points $(0,0), (S_m,0), (0,S_m)$ and $(S_m,S_m)$. The optimal solution of the problem is that of the LP which yields a minimum service cost. This solution can be extended for a general $T$. Figure \ref{Fig:Opt_Soln_Char_N=2_T>1} shows the SP optimal policy which is explained in the following proposition.
\begin{proposition}\label{Prop:Optimal_Solution_N=2_T=1}
    For $N=2$ and for each data content $m$, the SP optimal centralized caching policy follows: 
    \begin{enumerate}
        \item $(S_m,S_m)$ is optimal if and only if:
        \begin{equation*}
            \begin{aligned}
                r &< r_1 = \min_{i=1,2} \tau_i,\\
                0 &\leq \tau_i = \frac{1}{T} \sum\limits_{t=1}^{T} p_{i,t}^{m} \left( 1-\sum\limits_{l=1}^{L} \theta_{1,t}^{l} \theta_{2,t}^{l} \right) \leq 1            
            \end{aligned}
        \end{equation*}
        \item (0,0) is optimal if and only if: 
        \begin{equation*}
            \begin{aligned}
                r &> r_2 = \max_{i=1,2} \rho_i,\\
                \rho_i &= \frac{1}{T} \sum\limits_{t=1}^{T} \left( p_{i,t}^{m} + p_{j,t}^{m} \sum\limits_{l=1}^{L} \theta_{1,t}^{l} \theta_{2,t}^{l} \right) \geq 0
            \end{aligned}
        \end{equation*}
        \item $(S_m,0)$ is optimal if and only if $\tau_1 > \tau_2$.
        \item $(0,S_m)$ is optimal if and only if $\tau_2 > \tau_1$.
    \end{enumerate}
\begin{proof} The proof is straightforward by evaluating the cost function at all extreme points and comparing them to find the optimal solution.
\end{proof}
\end{proposition}
    \begin{figure}[h!]
        \centering
        \begin{tikzpicture}[scale=0.75]
            \draw[->,thick] (0,0) -- (16,0) node[anchor=north] {$r$};
            \draw	(0,0) node[anchor=north] {$0$}
                    (4,0) node[anchor=north] {$r_1$}
                    (11,0) node[anchor=north] {$r_2$}
                    (15,0) node[anchor=north] {$1$};
            \draw	(2,1) node[anchor=north] {$(S_m,S_m)$}
                	(7.5,2) node[anchor=north] {$(S_m,0) \Leftrightarrow \tau_1 > \tau_2$}
                	(7.5,1) node[anchor=north] {$(0,S_m) \Leftrightarrow \tau_1 < \tau_2$}
                	(13,1) node[anchor=north] {$(0,0)$};
            \draw [fill=red] (0,0) circle (0.1cm);
            \draw [fill=red] (4,0) circle (0.1cm);
            \draw [fill=red] (11,0) circle (0.1cm);
            \draw [fill=red] (15,0) circle (0.1cm);
        \end{tikzpicture}
        \caption{Optimal centralized caching policy for $N=2$}
        \label{Fig:Opt_Soln_Char_N=2_T>1}
    \end{figure}
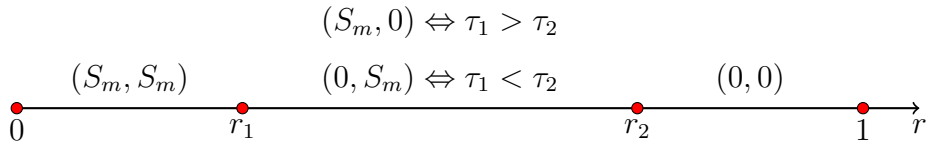

\subsubsection{\textbf{\underline{Case Study ($N=3$)}}}

For simplicity, we start by the case when $T=1$ and then extend it for any value of $T$. In this case, suffix $t$ can be dropped and the optimization problem will be:
    \begin{equation} \label{Eq:SP_Optimization_Problem_N=3_T=1}
        \begin{split}
            & \min \hspace{5mm} L^{\mathcal{P}} + r \sum_{m=1}^{M} \sum_{n=1}^{3} x_n^m\\
            & \text{s.t.} \hspace{5mm} 0 \leq x_n^m \leq S_m, \hspace{5mm} \forall n,m
        \end{split}
    \end{equation}
The problem decomposes to $M$ sub-problems and we have $1+\sum_{n=1}^{3} \binom{3}{n}=8$ sub-cases. Each sub-case leads to a different LP. The optimal solution is that of the one which yields a minimum service cost among all sub-cases. It is enough to consider only two sub-cases: (i) when all terms inside the $\bigl(.\bigr)^+$ functions are positive leading to a LP, where:
    \begin{equation}\label{Eq:Load_N=3_T=1_Case1}
        \begin{split}
            L^{\mathcal{P}} = \underbrace{\sum_{m=1}^{M} \sum_{n=1}^{3} S_m p_n^m}_\text{reactive load} - \underbrace{\sum_{m=1}^{M} \sum_{n=1}^{3} x_n^m p_n^m}_\text{caching gain} - \underbrace{\sum_{m=1}^{M} \sum_{i \neq j} \left( \Bigl(x_i^m p_j^m + x_j^m p_i^m \Bigr) \sum_{l=1}^{L} \theta_i^l \theta_j^l \right)}_\text{sharing gain}
        \end{split}
    \end{equation}
(ii) when all the terms inside the $\bigl(.\bigr)^+$ functions are negative and can be removed leading to another LP, where:
    \begin{equation}\label{Eq:Load_N=3_T=1_Case8}
        \begin{split}
            L^{\mathcal{P}} = \underbrace{\sum_{m=1}^{M} \sum_{n=1}^{3} S_m p_n^m}_\text{reactive load} - \underbrace{\sum_{m=1}^{M} \sum_{n=1}^{3} x_n^m p_n^m}_\text{caching gain}  - \underbrace{\sum_{m=1}^{M} \sum_{n=1}^{3} \Bigl( S_m - x_n^m \Bigr) p_n^m v_n}_\text{sharing gain}
        \end{split}
    \end{equation}
where,
\begin{equation*}
        v_i = \sum_{l=1}^{L} \biggl(\theta_i^l \theta_j^l+ \theta_i^l \theta_k^l \Bigl(1 - \theta_j^l \Bigl) \biggr), i \neq j \neq k
\end{equation*}

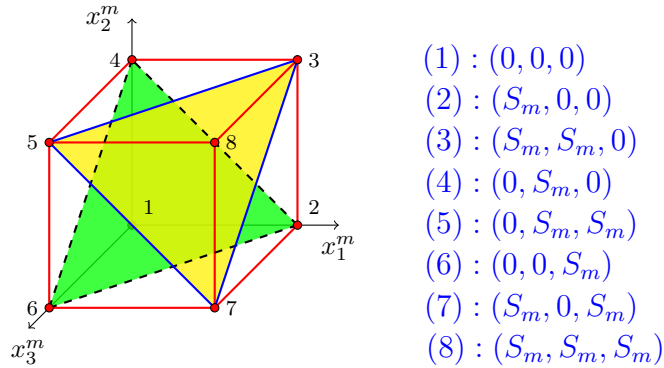
\begin{figure}[h!]
    \centering
    \begin{tikzpicture}[scale=0.55]
        \draw[->] (0,0) -- (5,0) node[anchor=north] {\small{$x_1^m$}};
        \draw[->] (0,0) -- (0,5) node[anchor=east] {\small{$x_2^m$}};
        \draw[->] (0,0) -- (-2.5,-2.5) node[anchor=north] {\small{$x_3^m$}};
        \coordinate (A1) at (0cm,0cm);
        \coordinate (A2) at (4cm,0cm);
        \coordinate (A3) at (4cm,4cm);
        \coordinate (A4) at (0cm,4cm);
        \coordinate (A5) at (-2cm,2cm);
        \coordinate (A6) at (-2cm,-2cm);
        \coordinate (A7) at (2cm,-2cm);
        \coordinate (A8) at (2cm,2cm);
        \draw [fill=red] (A1) circle (0.1cm);
    	\fill[fill=green,opacity=0.75] (A2) -- (A4) -- (A6) -- cycle;
    	\fill[fill=yellow,opacity=0.75] (A3) -- (A5) -- (A7) -- cycle;
    	\draw[thick,red] (A2) -- (A3);
    	\draw[thick,red] (A3) -- (A4);
    	\draw[thick,red] (A5) -- (A6);
    	\draw[thick,red] (A6) -- (A7);
    	\draw[thick,red] (A7) -- (A8);
    	\draw[thick,red] (A8) -- (A5);
    	\draw[thick,red] (A2) -- (A7); 
    	\draw[thick,red] (A3) -- (A8); 
    	\draw[thick,red] (A4) -- (A5); 
    	\draw[thick,dashed,black] (A2) -- (A4); 
    	\draw[thick,dashed,black] (A2) -- (A6); 
    	\draw[thick,dashed,black] (A4) -- (A6); 
    	\draw[thick,blue] (A3) -- (A5); 
    	\draw[thick,blue] (A5) -- (A7); 
    	\draw[thick,blue] (A7) -- (A3); 
        \draw [fill=red] (A2) circle (0.1cm);
        \draw [fill=red] (A3) circle (0.1cm);
        \draw [fill=red] (A4) circle (0.1cm);
        \draw [fill=red] (A5) circle (0.1cm);
        \draw [fill=red] (A6) circle (0.1cm);
        \draw [fill=red] (A7) circle (0.1cm);
        \draw [fill=red] (A8) circle (0.1cm);
        \draw (A1) node[anchor=south west] {\scriptsize{$1$}};
        \draw (A2) node[anchor=south west] {\scriptsize{$2$}};
        \draw (A3) node[anchor=west] {\scriptsize{$3$}};
        \draw (A4) node[anchor=east] {\scriptsize{$4$}};
        \draw (A5) node[anchor=east] {\scriptsize{$5$}};
        \draw (A6) node[anchor=east] {\scriptsize{$6$}};
        \draw (A7) node[anchor=west] {\scriptsize{$7$}};
        \draw (A8) node[anchor=west] {\scriptsize{$8$}};

        \draw[blue] (9,4) node{$(1): (0,0,0)$};
        \draw[blue] (9.35,3) node{$(2): (S_m,0,0)$};
        \draw[blue] (9.65,2) node{$(3): (S_m,S_m,0)$};
        \draw[blue] (9.35,1) node{$(4): (0,S_m,0)$};
        \draw[blue] (9.65,0) node{$(5): (0,S_m,S_m)$};
        \draw[blue] (9.35,-1) node{$(6): (0,0,S_m)$};
        \draw[blue] (9.65,-2) node{$(7): (S_m,0,S_m)$};
        \draw[blue] (10,-3) node{$(8): (S_m,S_m,S_m)$};

    \end{tikzpicture}
    \caption{Feasibility regions for $N=3$}
    \label{Fig:Feasibility_Regions_N=3}
\end{figure}

Note that the first term in (\ref{Eq:Load_N=3_T=1_Case1}), (\ref{Eq:Load_N=3_T=1_Case8}) represents the \emph{reactive load} of the network, the second term represents the \emph{caching gain} achieved by caching $x_1^m,x_2^m$ and $x_3^m$, while the last term represents the \emph{sharing gain} attained when each user transfer his cached data to other users. The feasibility regions of these LPs are shown in Figure \ref{Fig:Feasibility_Regions_N=3}. Each LP has 4 extreme points and its solution is one of them. The optimal solution of the problem is the solution of the LP which yields a minimum service cost. This solution can be extended for a general $T$. Figure \ref{Fig:Opt_Soln_Char_N=3_T>1} shows SP's optimal policy which is explained in the following proposition.
\begin{proposition}\label{Prop:Optimal_Solution_N=3_T=1}
For $N=3$ and for each data content $m$, the SP optimal centralized caching policy follows: 
\begin{enumerate}
    \item $(S_m,S_m,S_m)$ is optimal if and only if:
    \begin{equation*}
        \begin{aligned}
        r < r_1 = \min\limits_{i=1,2,3} \frac{1}{T} \sum\limits_{t=1}^{T} p_{i,t}^m \Bigl(1-v_{i,t}\Bigr), 0 \leq r_1 \leq 1
        \end{aligned}
    \end{equation*}
    \item Caching once is better than twice if and only if:
    \end{enumerate}
    \begin{equation*}
        \begin{aligned}
            r > r_2 = \max\limits_{i=1,2,3} \left\{ \max\limits_{k \neq j \neq i} \left\{ \frac{1}{T} \sum\limits_{t=1}^{T} \left( p_{j,t}^m \Bigl( 1 - \sum\limits_{l=1}^{L} \theta_{i,t}^l \theta_{j,t}^l \Bigr) + p_{k,t}^m \sum\limits_{l=1}^{L} \theta_{j,t}^l \theta_{k,t}^l \Bigl( 1 - \theta_{i,t}^l \Bigr) \right) \right\} \right\} 
        \end{aligned}
    \end{equation*}
    \begin{enumerate} \setcounter{enumi}{2}
    \item $(0,0,0)$ is optimal if and only if:
    \end{enumerate}
    \begin{equation*}
        \begin{aligned}
            r > & r_3 = \max\limits_{i=1,2,3} \frac{1}{T} \sum\limits_{t=1}^{T} \left( p_{i,t}^m + \sum\limits_{j \neq i} p_{j,t}^m \sum\limits_{l=1}^{L} \theta_{i,t}^l \theta_{j,t}^l \right) \geq 0 
        \end{aligned}
    \end{equation*}
    \begin{enumerate} \setcounter{enumi}{3}
    \item For the case of caching once, it is optimal to cache at user $k_1$ if and only if:
    \end{enumerate}
    \begin{equation} \label{Eq:Ranking_N=3_Once}
        k_1 = \argmax\limits_{i=1,2,3} \hspace{2mm} \frac{1}{T} \sum\limits_{t=1}^{T} \left( p_{i,t}^m + \sum\limits_{j \neq i} \left( p_{j,t}^m \sum\limits_{l=1}^{L} \theta_{i,t}^l \theta_{j,t}^l \right) \right)
    \end{equation}
    \begin{enumerate} \setcounter{enumi}{4}
    \item For the case of caching twice, it is optimal to cache at users $k_1$ and $k_2$ if and only if:
    \end{enumerate}
    \begin{equation} \label{Eq:Ranking_N=3_Twice}
        \begin{split}
            (k_1,k_2) & = \argmax\limits_{\substack{i,j=1,2,3 \\ i\neq j}} \hspace{2mm} \frac{1}{T} \sum\limits_{t=1}^{T} \left( p_{i,t}^m + p_{j,t}^m + \sum\limits_{k \neq i,j} p_{k,t}^m v_{k,t} \right)
        \end{split}
    \end{equation}        
\begin{proof} The proof is straight forward by evaluating the cost function for every extreme point and comparing them to find the optimal solution.
\end{proof}
\end{proposition}
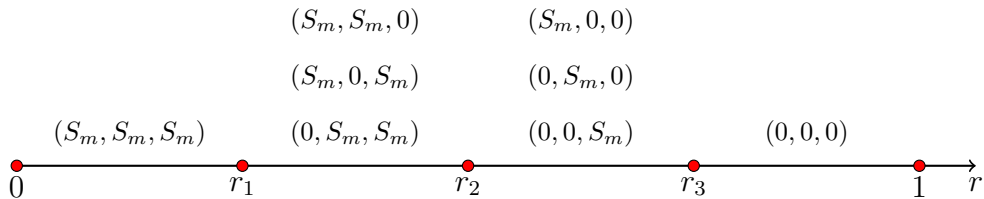
\begin{figure}[h!]
    \centering
    \begin{tikzpicture}[scale=0.75]
        \draw[->,thick] (0,0) -- (17,0) node[anchor=north] {$r$};
        \draw	(0,0) node[anchor=north] {$0$}
                (4,0) node[anchor=north] {$r_1$}
                (8,0) node[anchor=north] {$r_2$}
                (12,0) node[anchor=north] {$r_3$}
                (16,0) node[anchor=north] {$1$};
        \draw	(2,1) node[anchor=north] {\small{$(S_m,S_m,S_m)$}}
                (6,3) node[anchor=north] {\small{$(S_m,S_m,0)$}}
                (6,2) node[anchor=north] {\small{$(S_m,0,S_m)$}}
                (6,1) node[anchor=north] {\small{$(0,S_m,S_m)$}}
                (10,3) node[anchor=north] {\small{$(S_m,0,0)$}}
                (10,2) node[anchor=north] {\small{$(0,S_m,0)$}}
                (10,1) node[anchor=north] {\small{$(0,0,S_m)$}}
                (14,1) node[anchor=north] {\small{$(0,0,0)$}};
        \draw [fill=red] (0,0) circle (0.1cm);
        \draw [fill=red] (4,0) circle (0.1cm);
        \draw [fill=red] (8,0) circle (0.1cm);
        \draw [fill=red] (12,0) circle (0.1cm);
        \draw [fill=red] (16,0) circle (0.1cm);
    \end{tikzpicture}
    \caption{Optimal centralized caching policy for $N=3$}
    \label{Fig:Opt_Soln_Char_N=3_T>1}
\end{figure}

\subsubsection{\textbf{\underline{Optimal Policy for $N$-users}}}

Now, from the previous cases, we can infer the optimal centralized caching policy for a general number of users $N$ as shown in Figure \ref{Fig:Opt_Soln_Char_N_T>1}. The optimal solution depends on the value of $r$. For each data content $m$, SP needs to find $N$ points on the scale or $r$. In particular, there are $N+1$ regions on this scale starting from caching nothing up to caching everywhere. These points (regions) determine how much caching should be done based on the exact value of $r$. Inside each region of $r$, users should be ranked to determine which user caches this data content. Based on the demand and mobility profiles of all users, SP finds an optimal caching decision $\{x_n^{m*}\}_n^m$ which minimizes the incurred service cost. We summarize the optimal caching policy in Algorithm \ref{Alg:Optimal_Centralized_Policy}.
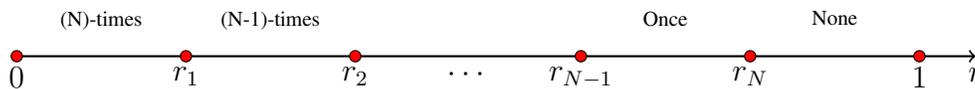
\begin{figure}[h!]
    \centering
    \begin{tikzpicture}[scale=0.75]
        \draw[->,thick] (0,0) -- (17,0) node[anchor=north] {$r$};
        \draw	(0,0) node[anchor=north] {$0$}
    		    (3,0) node[anchor=north] {$r_1$}
    		    (6,0) node[anchor=north] {$r_2$}
    		    (8,0) node[anchor=north] {$\cdots$}
    		    (10,0) node[anchor=north] {$r_{N-1}$}
    		    (13,0) node[anchor=north] {$r_N$}
    		    (16,0) node[anchor=north] {$1$};
        \draw	(1.5,1) node[anchor=north] {\scriptsize{(N)-times}}
    		    (4.5,1) node[anchor=north] {\scriptsize{(N-1)-times}}
    		    (11.5,1) node[anchor=north] {\scriptsize{Once}}
    		    (14.5,1) node[anchor=north] {\scriptsize{None}};
        \draw [fill=red] (0,0) circle (0.1cm);
        \draw [fill=red] (3,0) circle (0.1cm);
        \draw [fill=red] (6,0) circle (0.1cm);
        \draw [fill=red] (10,0) circle (0.1cm);
        \draw [fill=red] (13,0) circle (0.1cm);
        \draw [fill=red] (16,0) circle (0.1cm);
    \end{tikzpicture}
    \caption{Optimal centralized caching policy for $N$ users}
    \label{Fig:Opt_Soln_Char_N_T>1}
\end{figure}
\begin{remark}\label{Rk:Remark_N_1}
The optimal caching decision depends on the exact value of $r$ which represents the caching cost. In particular, smaller value of $r$ yields more caching and vice verse.
\end{remark}
\begin{remark}
Higher values of users meeting probabilities shift points $r_1,r_2,\cdots,r_{N-1}$ to the left while point $r_N$ moves to the right. In particular, the possibility of over-caching reduces when the meeting probabilities increase. Moreover, higher demand levels shift all points to the right and yields more caching. 
\end{remark}

\begin{proposition}
    The complexity of the SP optimal centralized caching policy described in Algorithm \ref{Alg:Optimal_Centralized_Policy} grows exponentially with the number of users $N$.
\begin{proof}Although the number of points on the scale of $r$ increases linearly with the number of users $N$, the total number of required terms for users ranking is given by:
    \begin{equation*}
        \begin{split}
            &\binom{N}{1}+\binom{N}{2}+\cdots+\binom{N}{N-2}+\binom{N}{N-1} \\
            &= \sum\limits_{k=0}^{N} \binom{N}{k} - \binom{N}{N} - \binom{N}{0} = 2^N - 2
        \end{split}
    \end{equation*}
which increases exponentially with $N$.
\end{proof}
\end{proposition}
\begin{algorithm}
    \caption{Optimal Centralized Caching Policy for $N$-users}
    \label{Alg:Optimal_Centralized_Policy}
    \begin{algorithmic}
    \State \textbf{Given:} $N,M,L,T,\vec{\Pi}_n,\vec{\Theta}_{n}$
    \For{$m=1$ to $M$}
        \Call{Caching}{$N,L,T,m,p_{n,t}^{m},\theta_{n,t}^{l}$}
    \EndFor
        \Procedure{Caching}{$N,L,T,m,p_{n,t}^{m},\theta_{n,t}^{l}$}
            \State \textbf{\underline{Step 1:}}
            \State Create a level-$(1)$ ranked list using: 
                \begin{equation}\label{Eq:Ranking_Optimal_N_1}
                    s_i = \frac{1}{T} \sum\limits_{t=1}^{T} \left( p_{i,t}^m+\sum\limits_{j \neq i} \biggl( p_{j,t}^m \sum\limits_{l=1}^{L} \theta_{i,t}^l \theta_{j,t}^l \biggr) \right), \forall i
                \end{equation}
            \State $r_N$ is the first item in list, $k_1$ is its corresponding index.
            \State User $k_1$ cache it while all other users will not.
            \State \textbf{\underline{Step 2:}}
            \State Create a level-(2) ranked list using:
                \begin{equation}\label{Eq:Ranking_Optimal_N_2}
                    \begin{aligned}
                        s_{ij} = \frac{1}{T} \sum\limits_{t=1}^{T} \Biggl( p_{i,t}^m + p_{j,t}^m + \sum\limits_{k \neq i,j} \biggl( p_{k,t}^m \sum\limits_{l=1}^{L} \Bigl( \theta_{i,t}^l \theta_{k,t}^l + \theta_{j,t}^l \theta_{k,t}^l - \theta_{i,t}^l \theta_{j,t}^l \theta_{k,t}^l \Bigr) \biggr) \Biggr), \forall i,j
                    \end{aligned}
                \end{equation}
            \State $k_1,k_2$ are the corresponding indexes of the first item.
			 \State Users $k_1$ and $k_2$ cache it while others will not.
            \State calculate $r_{N-1}$ as follows:
                \begin{equation}\label{Eq:Optimal_N_r_N-1}
                    \begin{aligned}
                        r_{N-1} = \frac{1}{T} \sum\limits_{t=1}^{T} \Biggl( p_{k_2,t}^m \Bigl(1 - \sum\limits_{l=1}^{L} \theta_{k_1,t}^l \theta_{k_2,t}^l \Bigr) + \sum\limits_{n \neq k_1,k_2} p_{n,t}^m \sum\limits_{l=1}^{L} \theta_{k_2,t}^l \theta_{n,t}^l \Bigl(1 - \theta_{k_1,t}^l \Bigr) \Biggr)
                    \end{aligned}
                \end{equation}
            \State $\hspace{1cm} \vdots$
            \State \textbf{\underline{Step $N-1$:}}
            \State Create a level-$(N-1)$ ranked list using 
                \begin{equation}\label{Eq:Ranking_Optimal_N_N-1}
                    \begin{aligned}
                    s_i &= \frac{1}{T} \sum\limits_{t=1}^{T} \biggl( p_{i,t}^m \Bigl( 1 - v_{i,t} \Bigr) \biggr), \forall i,\\
						v_{i,t} &= \sum\limits_{l=1}^{L} \left[ \sum\limits_{i \neq j} \theta_{i,t}^l \theta_{j,t}^l -  \sum\limits_{i \neq j \neq k} \theta_{i,t}^l \theta_{j,t}^l \theta_{k,t}^l \cdots (-1)^{N} \prod_{n=1}^{N} \theta_{n,t}^l \right]
                    \end{aligned}
                \end{equation}
            \State $r_1$ is the smallest item in list, $k_N$ is the corresponding index.
            \State User $k_N$ will not receive a caching assignment and all other users cache it.
        \EndProcedure
    \end{algorithmic}
\end{algorithm}

\subsection{Greedy Centralized Caching Policy}

The complexity of optimal policy discussed in the previous section motivates us to introduce a greedy algorithm as a sub-optimal caching policy. Moreover, this algorithm allows us to establish upper and lower bounds on the achieved service gain of the optimal caching policy. The main idea of this algorithm depends on ranking users based on the level-$(1)$ list stated in (\ref{Eq:Ranking_Optimal_N_1}). Users are picked from this list in order to cache the amount required of each data content. Algorithm \ref{Alg:Suboptimal_Centralized_Policy} summarizes the steps of this greedy algorithm. 
\begin{algorithm}[h!]
    \caption{Greedy Centralized Caching Policy for $N$-users}
    \label{Alg:Suboptimal_Centralized_Policy}
    \begin{algorithmic}
    \State \textbf{Given:} $N,M,L,T,\vec{\Pi}_n,\vec{\Theta}_{n}$
    \For{$m=1$ to $M$}
        \Call{Caching}{$N,L,T,m,p_{n,t}^{m},\theta_{n,t}^{l}$}
    \EndFor
        \Procedure{Caching}{$N,L,T,m,p_{n,t}^{m},\theta_{n,t}^{l}$}
            \State \textbf{\underline{Step 1:}}
            \State Create a level-$(1)$ ranked list using: 
                \begin{equation*}\label{Eq:Ranking_Suboptimal_N_1}
                    s_i = \frac{1}{T} \sum\limits_{t=1}^{T} \left( p_{i,t}^m+\sum\limits_{j \neq i} \biggl( p_{j,t}^m \sum\limits_{l=1}^{L} \theta_{i,t}^l \theta_{j,t}^l \biggr) \right), \forall i
                \end{equation*}
            \State $r_N$ is the first item in list, $k_1$ is its corresponding index.
            \State \textbf{\underline{Step 2:}}
            \State Exclude user $k_1$ and create a level-$(2)$ ranked list using:
                \begin{equation}\label{Eq:Ranking_Suboptimal_N_2}
                    \begin{aligned}
                        s_j = \frac{1}{T} \sum\limits_{t=1}^{T} \Biggl( p_{j,t}^m \Bigl(1 - \sum\limits_{l=1}^{L} \theta_{k_1,t}^l \theta_{j,t}^l \Bigr) + \sum\limits_{k \neq j} p_{k,t}^m \sum\limits_{l=1}^{L} \theta_{j,t}^l \theta_{k,t}^l \Bigl(1 - \theta_{k_1,t}^l \Bigr) \Biggr), \forall j \neq k_1
                    \end{aligned}
                \end{equation}
            \State $r_{N-1}$ is the first item in list.
            \State $k_2$ is the corresponding index.
            \State \textbf{\underline{Step 3:}}
            \State Exclude users $k_1,k_2$ and create a level-$(3)$ ranked list by:
                \small{
                \begin{equation}\label{Eq:Ranking_Suboptimal_N_3}
                    \begin{aligned}
                    s_k &= \frac{1}{T} \sum\limits_{t=1}^{T} \Biggl( p_{k,t}^m \biggl[ 1 - \sum\limits_{l=1}^{L} \Bigl( \theta_{k_1,t}^l \theta_{k,t}^l + \theta_{k_2,t}^l \theta_{k,t}^l \bigl( 1 - \theta_{k_1,t}^l \bigr) \Bigr) \biggr] \\
                    &+ \sum\limits_{n \neq k} p_{n,t}^m \sum\limits_{l=1}^{L} \theta_{k,t}^l \theta_{n,t}^l \Bigl( 1 - \theta_{k_1,t}^l \Bigr) \Bigl( 1 - \theta_{k_2,t}^l \Bigr) \Biggr), \forall k \neq k_1,k_2
                    \end{aligned}
                \end{equation}}
            \State $r_{N-2}$ is the first item in list, $k_3$ is the corresponding index.
            \State $\hspace{1cm} \vdots$
            \State So on for all other ranges of $r$.
        \EndProcedure
    \end{algorithmic}
\end{algorithm}
\begin{proposition}
    Complexity of the greedy caching policy described in Algorithm \ref{Alg:Suboptimal_Centralized_Policy} grows in a polynomial order with the number of users $N$.
\begin{proof} We still need to find $N$ points the scale of $r$ which increases linearly with the number of users. However, the total number of required terms is given by:
    \begin{equation*}
        1+2+\cdots+(N-2)+(N-1)+N = \frac{N(N+1)}{2}
    \end{equation*}    
which grows in a polynomial order with $N$.
\end{proof}
\end{proposition}

\subsection{Upper and Lower Bounds Analysis}

The greedy caching policy ranks users based on the level-$(1)$ list stated in  (\ref{Eq:Ranking_Optimal_N_1}). This ranking is similar to that of caching once in the optimal policy. However, this doesn't guarantee that level-$(1)$ ranking is still valid in all other cases. For instance, the first two users in this list are not guaranteed to be the same users for the case of caching twice in the optimal policy. Hence, the greedy caching policy forms a lower bound for the proactive service gain achieved by the optimal policy.
\begin{theorem}
    Under demand and mobility profiles of $N$-users and for $T \geq 1$, the optimal proactive service gain $\bigtriangleup C\left(\vec{\Pi}_n,\vec{\Theta}_n\right)$ of (\ref{Eq:SP_Optimization_Problem}) achieved by Algorithm (\ref{Alg:Optimal_Centralized_Policy}) satisfies:
        \begin{equation}\label{Eq:Gain_Lower_Bound}
            \bigtriangleup C\left(\vec{\Pi}_n,\vec{\Theta}_n\right) \geq \bigtriangleup C_{L}\left(\vec{\Pi}_n,\vec{\Theta}_n\right)
        \end{equation}  
    where, $\bigtriangleup C_{L}\left(\vec{\Pi}_n,\vec{\Theta}_n\right)$ is the gain achieved by Algorithm (\ref{Alg:Suboptimal_Centralized_Policy}).
\begin{proof} We compare the gain achieved by the optimal and greedy policies. For example, in the case of caching once both policies achieve the same gain and we have:
    \begin{equation}\label{Eq:Gain_Lower_Bound_Once}
        \begin{aligned}
            \bigtriangleup C\left(\vec{\Pi}_n,\vec{\Theta}_n\right) = \bigtriangleup C_{L}\left(\vec{\Pi}_n,\vec{\Theta}_n\right)= \sum\limits_{m=1}^{M} S_m \left[\frac{1}{T} \sum\limits_{t=1}^{T} \left( p_{i,t}^m+\sum\limits_{j \neq i} \biggl( p_{j,t}^m \sum\limits_{l=1}^{L} \theta_{i,t}^l \theta_{j,t}^l \biggr) \right) -r \right]
        \end{aligned}
    \end{equation}
In the case of caching twice, the proactive service gain of the optimal caching policy is:
    \begin{equation}\label{Eq:Gain_Lower_Bound_Twice}
        \begin{aligned}
            \bigtriangleup C & \left(\vec{\Pi}_n,\vec{\Theta}_n\right) = \\  &\sum\limits_{m=1}^{M} S_m \left[\frac{1}{T} \sum\limits_{t=1}^{T} \Biggl( p_{i,t}^m + p_{j,t}^m + \sum\limits_{k \neq i,j} p_{k,t}^m \sum\limits_{l=1}^{L} \biggl( \theta_{i,t}^l \theta_{k,t}^l + \theta_{j,t}^l \theta_{k,t}^l - \theta_{i,t}^l \theta_{j,t}^l \theta_{k,t}^l \biggr) \Biggr) -2r \right]
        \end{aligned}
    \end{equation}
This gain depends mainly on the selection of users $i$ and $j$. Since level-$(1)$ ranking can not guarantee that these users are the same users as in the greedy policy, this gain is larger than or equal to the gain achieved by the greedy caching policy. The same approach applies to show a similar result for all other cases.
\end{proof}
\end{theorem}
Moreover, the greedy algorithm allows us to establish an upper bound for the optimal proactive service gain. Level-$(1)$ ranking defined in (\ref{Eq:Ranking_Optimal_N_1}) generates $N$ items representing the gain achieved by caching data content $m$ once at one of the users. Adding these gains up provides us with an upper bound for the gain achieved by the optimal caching policy in all cases. For example, the first item is this ranked list is an upper bound for the case of caching once in the optimal policy. The sum of the first two items is an upper bound for the case of caching twice in the optimal policy and so on. This result is stated in the following theorem.
\begin{theorem}
    Under demand and mobility profiles of $N$-users and for $T \geq 1$, the optimal proactive service gain $\bigtriangleup C\left(\vec{\Pi}_n,\vec{\Theta}_n\right)$ of (\ref{Eq:SP_Optimization_Problem}) achieved by Algorithm (\ref{Alg:Optimal_Centralized_Policy}) satisfies:
        \begin{equation}\label{Eq:Gain_Upper_Bound}
            \bigtriangleup C\left(\vec{\Pi}_n,\vec{\Theta}_n\right) \leq \bigtriangleup C_{U}\left(\vec{\Pi}_n,\vec{\Theta}_n\right)
        \end{equation}  
    where, $\bigtriangleup C_{U}\left(\vec{\Pi}_n,\vec{\Theta}_n\right)$ is the gain achieved by adding up gains defined in (\ref{Eq:Ranking_Optimal_N_1}).
\begin{proof} We show this result by comparing the gain achieved by the optimal caching policy with the again achieved by the greedy caching policy by adding up items of level-$(1)$ ranking list (\ref{Eq:Ranking_Optimal_N_1}). For the case of caching once we have:
    \begin{equation}\label{Eq:Gain_Upper_Bound_Once}
        \bigtriangleup C\left(\vec{\Pi}_n,\vec{\Theta}_n\right) = \bigtriangleup C_{U}\left(\vec{\Pi}_n,\vec{\Theta}_n\right)
    \end{equation} 
which is the same value as in (\ref{Eq:Gain_Lower_Bound_Once}). For the case of caching twice, we have:
    \begin{equation}\label{Eq:Gain_Upper_Bound_Twice}
        \begin{aligned}
            \bigtriangleup C_{U}\left(\vec{\Pi}_n,\vec{\Theta}_n\right) &= \sum\limits_{m=1}^{M} S_m \left[\frac{1}{T} \sum\limits_{t=1}^{T} \Biggl( \Bigl(p_{i,t}^m + p_{j,t}^m\Bigr)\Bigl( 1 + \sum\limits_{l=1}^{L} \theta_{i,t}^l \theta_{j,t}^l \Bigr) \right. \\
            &\left. + \sum\limits_{k \neq i,j} \biggl( p_{k,t}^m \sum\limits_{l=1}^{L} \Bigl( \theta_{i,t}^l \theta_{k,t}^l + \theta_{j,t}^l \theta_{k,t}^l \Bigr) \biggr) \Biggr) -2r \right]
        \end{aligned}
    \end{equation}
Comparing (\ref{Eq:Gain_Upper_Bound_Twice}) with (\ref{Eq:Gain_Lower_Bound_Twice}) we see that:
    \begin{equation}\label{Eq:Gain_Upper_Bound_Twice_2}
        \bigtriangleup C_{U} \left(\vec{\Pi}_n,\vec{\Theta}_n \right) \geq \bigtriangleup C\left(\vec{\Pi}_n,\vec{\Theta}_n\right)
    \end{equation}
Same approach applies to all other cases.
\end{proof}
\end{theorem}

\subsection{Choosing Optimal Reward}\label{Sec:Choosing_Reward}

The SP service cost consists of two components, one component represents the corresponding cost for serving the peak load, and another component representing the caching cost affected by the reward value $r$. In the previous section, we introduced the solution of the centralized caching scheme based on the reward value $r$. The SP pays this reward back to users to incentivize them to participate in this model. The aim of the SP is to minimize this reward as much as possible to reduce its expected cost. On the other hand, users try to get as much reward as they can to minimize their expected payment. This creates a tension between the SP and users and raises the question of what is the optimal rewarding value on which both SP and users will agree. To see this we assume that each user has an isolated memory $Z_n$. Suppose, for simplicity, that all users have the same memory capacity of $Z_n$. Therefore, SP can see an aggregate memory of size $Z=NZ_n$. Users have a reward preference to assign a certain memory of their devices for this caching process. We consider a linear relation between the reward $(r)$ and the memory size $(Z)$, i.e. $Z(r)=\beta r$, for some $\beta>0$. In particular, users assign more memory for larger values of the reward and vice versa. Now, SP aims to let users assign enough memory for caching but this memory should not be more than required. We can see this by rewriting (\ref{Eq:SP_Optimization_Problem}) as follows:
    \begin{equation}\label{Eq:SP_Optimization_Problem_Modified}
        \begin{aligned}
            & \min \hspace{5mm} L_t^{\mathcal{P}}\\
            & \text{s.t.} \hspace{5mm} \sum_{n=1}^{N} \sum_{m=1}^{M} x_n^m \leq Z.
        \end{aligned}
    \end{equation}
Converting this problem to an unconstrained problem, we get
    \begin{equation}\label{Eq:SP_Optimization_Problem_Modified2}
        \begin{aligned}
            \underset{x_n^m}{\min} \hspace{3mm} L_t^{\mathcal{P}} + r\left(\sum_{n=1}^{N} \sum_{m=1}^{M} x_n^m - Z\right) =  \underset{x_n^m}{\min} \hspace{3mm} \left( L_t^{\mathcal{P}} + r \sum_{n=1}^{N} \sum_{m=1}^{M} x_n^m \right)- rZ 
        \end{aligned}
    \end{equation}
where $r$ is the largrange multiplier corresponding to the constraint $\sum\limits_{n=1}^{N} \sum\limits_{m=1}^{M} x_n^m \leq Z$. Notice that the first term in (\ref{Eq:SP_Optimization_Problem_Modified2}) is the problem we solved in Section \ref{Sec:Centralized_Optimal} without having this memory constraint. The SP solves this optimization problem for each value of Z. The optimal choice of $r$ depends on users preference. So, plotting the Lagrangian multiplier $r$ for the solution of (\ref{Eq:SP_Optimization_Problem_Modified2}) versus $Z$ we get the curve shown in Figure \ref{Fig:Choosing_Reward}. The optimal reward  choice $r^*$ corresponds to a memory $Z^*$ and this point is determined by the intersection of the SP reward preference and the users reward preference. This shows that we can find a unique optimal solution for the problem discussed in Section \ref{Sec:Centralized_Optimal} by knowing $Z(r)$. Different functions of $Z(r)$ leads to another optimal solution. Therefore, one of the parameters that SP needs to learn about users is their preference function $Z(r)$.

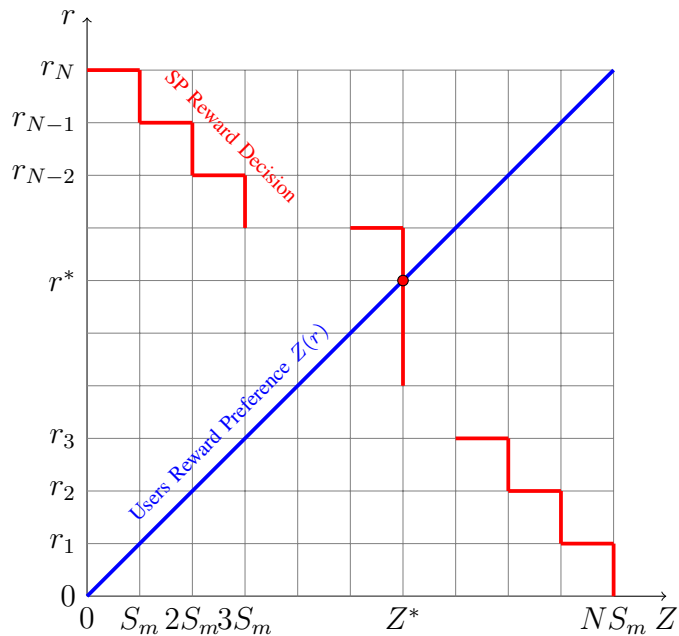
\begin{figure}[h!]
	\centering
      \begin{tikzpicture}[scale=0.70]
            
            \draw[->] (0,0) -- (11,0) node[anchor=north] {$Z$};
            \draw	(0,0) node[anchor=north] {$0$}
				(1,0) node[anchor=north] {$S_m$}
				(2,0) node[anchor=north] {$2S_m$}
				(3,0) node[anchor=north] {$3S_m$}
				(6,0) node[anchor=north] {\textbf{$Z^*$}}
                	(10,0) node[anchor=north] {$NS_m$};
            \draw[->] (0,0) -- (0,11) node[anchor=east] {$r$};
            \draw	(0,0) node[anchor=east] {$0$}
				(0,1) node[anchor=east] {$r_1$}
				(0,2) node[anchor=east] {$r_2$}
				(0,3) node[anchor=east] {$r_3$}
				(0,6) node[anchor=east] {\textbf{$r^*$}}
				(0,8) node[anchor=east] {$r_{N-2}$}
				(0,9) node[anchor=east] {$r_{N-1}$}
                	(0,10) node[anchor=east] {$r_N$};           
                
            \draw[step=1cm,gray,very thin] (0,0) grid (10,10);            

            \draw[line width=0.5mm,red] (0,10) -- (1,10);
            \draw[line width=0.5mm,red] (1,10) -- (1,9);
            \draw[line width=0.5mm,red] (1,9) -- (2,9);
            \draw[line width=0.5mm,red] (2,9) -- (2,8);
            \draw[line width=0.5mm,red] (2,8) -- (3,8);
            \draw[line width=0.5mm,red] (3,8) -- (3,7);
            \draw[line width=0.5mm,red] (5,7) -- (6,7);
            \draw[line width=0.5mm,red] (6,7) -- (6,4);
            \draw[line width=0.5mm,red] (7,3) -- (8,3);
            \draw[line width=0.5mm,red] (8,3) -- (8,2);
            \draw[line width=0.5mm,red] (8,2) -- (9,2);
            \draw[line width=0.5mm,red] (9,2) -- (9,1);
            \draw[line width=0.5mm,red] (9,1) -- (10,1);
            \draw[line width=0.5mm,red] (10,1) -- (10,0);

            \draw[line width=0.5mm,blue] (0,0) -- (10,10); 

            \draw [fill=red] (6,6) circle (0.1cm);

            \draw (3,9) node[red,rotate=-45,anchor=north] {\scriptsize{SP Reward Decision}};
            \draw  (3,3) node[blue,rotate=45,anchor=south] {\scriptsize{Users Reward Preference $Z(r)$}};
              
        \end{tikzpicture}	
	\caption{Choosing an optimal reward ($r$) for $N$ users in the centralized caching scheme}
	\label{Fig:Choosing_Reward}
\end{figure}

\subsection{Impact of User Mobility}\label{Sec:Mobility_Impact_Centralized}

User's mobility statistics affect the meeting probability between users. Going back to Algorithm \ref{Alg:Optimal_Centralized_Policy}, and the final result shown in Figure \ref{Fig:Opt_Soln_Char_N_T>1}, we can see that these meeting probabilities shift the points $r_1,r_2,\cdots,r_N$ and change the SP decision. 
\begin{definition}
The average meeting probability between any two users $i$ and $j$ is defined by:
\begin{equation}
    \alpha_{ij} = \frac{1}{T} \sum_{t=1}^{T} \sum_{l=1}^{L} \theta_{i,t}^{l} \theta_{j,t}^{l}
\end{equation}
\end{definition}
When users are moving such that they are meeting each other with a higher probability, i.e. the terms $\alpha_{ij}$ increase for all $i,j$, the points $r_1,r_2,\cdots,r_{N-1}$ shift to the left while the point $r_N$ shifts to the right. This means that the possibility of over-caching of this data content decreases and the SP optimal decision tends to be caching this content once. When the meeting probabilities decrease, the points $r_1,r_2,\cdots,r_{N-1}$ shift to the right. This means that the possibility of over-caching increases and the SP optimal tends to be caching this data content at multiple users.

To show how user's mobility affects the SP caching decision, we consider the case when users have similar mobility statistics. In particular, we consider the case when $\theta_{1,t}^{l}=\theta_{2,t}^{l}=\cdots=\theta_{N,t}^{l}=\theta_{t}^{l},\forall l \in \{1,2,\cdots,L\}$. We also consider the case when users move uniformly over all locations, i.e. $\theta_{t}^{l}=\frac{1}{L}, \forall t \in \{1,2,\cdots,T\}$. We keep the assumption that each user may have a different interest in each data content. This similarity in users mobility simplifies the optimal caching policy described in Algorithm \ref{Alg:Optimal_Centralized_Policy}. In particular, for any level-$(n)$, we create all possible $n$-tuples indices $a_n=\{i_1,i_2,\cdots,i_n\}$ and choose the optimal set of users with indices $a_n^*$ which satisfies:
\begin{equation}\label{Eq:Ranking_Optimal_Similarity}
    a_n^* = \underset{a_n}{\argmax} \hspace{5mm} \frac{1}{T}\sum_{t=1}^{T} \left( \sum_{i\in a_n} p_{i,t}^m + \left[ \sum_{k=1}^{n} (-1)^{k+1} \binom{n}{k} \sum_{l=1}^{L} \frac{1}{L^{k+1}} \right] \sum_{j \notin a_n} p_{j,t}^m \right)
\end{equation}
For example, at level-$(N-1)$, we have
\begin{equation}
    v_{n,t} = \sum_{k=1}^{N-1} (-1)^{k+1} \binom{N-1}{k} \sum_{l=1}^{L} \frac{1}{L^{k+1}}, \forall t
\end{equation}
and in this case, we choose the set of users by
\begin{equation}
    \begin{aligned}
    a_{N-1}^* &= \underset{a_{N-1}}{\argmax} \hspace{5mm} \frac{1}{T}\sum_{t=1}^{T} \left( \sum_{i\in a_{N-1}} p_{i,t}^m + p_{j,t}^m v_{j,t} \right), j\notin a_{N-1} \\
    &= \underset{a_{N-1}}{\argmax} \hspace{5mm} \frac{1}{T}\sum_{t=1}^{T} \left( \sum_{i=1}^{N} p_{i,t}^m - p_{j,t}^m \bigl(1-v_{j,t}\bigr) \right), j\notin a_{N-1} \\
    &= \frac{1}{T}\sum_{t=1}^{T} \sum_{i=1}^{N} p_{i,t}^m + \underset{j \in \mathcal{N}}{\argmin} \hspace{5mm} \frac{1}{T}\sum_{t=1}^{T} p_{j,t}^m \bigl(1-v_{j,t}\bigr).
    \end{aligned}
\end{equation}
Moreover, the greedy policy, described in Algorithm \ref{Alg:Suboptimal_Centralized_Policy}, becomes more simpler when we consider this similarity in users mobility. For each level-$(n)$ where $n\geq2$, knowing the set of selected users from the previous step $a_{n-1}$, we need to find the user $j^*$. This user will be added to the set $a_{n-1}$ to construct the new set $a_n$. This can be found by \begin{equation}\label{Eq:Ranking_SubOptimal_Similarity}
    \underset{j\notin a_{n-1}}{\argmax} \hspace{5mm} \frac{1}{T}\sum_{t=1}^{T} \left( p_{j,t}^m \left[1- \sum_{k=1}^{n-1} (-1)^{k+1} \binom{n-1}{k} \sum_{l=1}^{L} \frac{1}{L^{k+1}} \right] + \sum_{l=1}^{L} \frac{(L-1)^{n-1}}{L^{n+1}} \underset{\substack{k \neq j \\ k \notin a_{n-1}}}{\sum} p_{k,t}^m \right)
\end{equation}
Notice that (\ref{Eq:Ranking_SubOptimal_Similarity}), can be used to determine the points $r_1,r_2,\cdots,r_N$ for the optimal caching policy as well. We also notice that as $L\rightarrow\infty$, we have $\sum_{l=1}^{L} \frac{1}{L^{k+1}} \rightarrow 0$ and $\sum_{l=1}^{L} \frac{(L-1)^{n-1}}{L^{n+1}} \rightarrow 1$. Consequently, $\alpha_{ij} \rightarrow 0$ and $v_n \rightarrow 0$. Therefore, in this special case, the optimal caching policy depends only on users interest, as shown in Figure \ref{Fig:Opt_Soln_Similarity}, where $\hat{p}_{n}^m = \frac{1}{T} \sum_{t=1}^{T} p_{n,t}^m$.

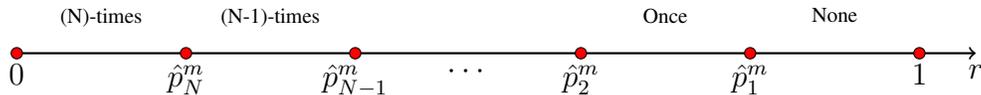
\begin{figure}[h!]
    \centering
    \begin{tikzpicture}[scale=0.75]
        \draw[->,thick] (0,0) -- (17,0) node[anchor=north] {$r$};
        \draw	(0,0) node[anchor=north] {$0$}
    		    (3,0) node[anchor=north] {$\hat{p}_N^m$}
    		    (6,0) node[anchor=north] {$\hat{p}_{N-1}^m$}
    		    (8,0) node[anchor=north] {$\cdots$}
    		    (10,0) node[anchor=north] {$\hat{p}_2^m$}
    		    (13,0) node[anchor=north] {$\hat{p}_1^m$}
    		    (16,0) node[anchor=north] {$1$};
        \draw	(1.5,1) node[anchor=north] {\scriptsize{(N)-times}}
    		    (4.5,1) node[anchor=north] {\scriptsize{(N-1)-times}}
    		    (11.5,1) node[anchor=north] {\scriptsize{Once}}
    		    (14.5,1) node[anchor=north] {\scriptsize{None}};
        \draw [fill=red] (0,0) circle (0.1cm);
        \draw [fill=red] (3,0) circle (0.1cm);
        \draw [fill=red] (6,0) circle (0.1cm);
        \draw [fill=red] (10,0) circle (0.1cm);
        \draw [fill=red] (13,0) circle (0.1cm);
        \draw [fill=red] (16,0) circle (0.1cm);
    \end{tikzpicture}
    \caption{Optimal centralized caching policy for $N$ users with a similar and uniformly distributed mobility patterns over $L\rightarrow\infty$ locations.}
    \label{Fig:Opt_Soln_Similarity}
\end{figure}

Now, considering the memory constraint and the reward preference, discussed in Section \ref{Sec:Choosing_Reward}, we will have the result shown in Figure \ref{Fig:Choosing_Reward_Similarity}. We notice that the optimal reward choice will depend only on users interest. Moreover, all users have interest less than $0.5$ will not receive any caching assignment. 

\begin{figure}[h!]
	\centering
      \begin{tikzpicture}[scale=0.70]
            
            \draw[->] (0,0) -- (11,0) node[anchor=north] {$Z$};
            \draw	(0,0) node[anchor=north] {$0$}
				(1,0) node[anchor=north] {$S_m$}
				(2,0) node[anchor=north] {$2S_m$}
				(3,0) node[anchor=north] {$3S_m$}
				(6,0) node[anchor=north] {\textbf{$Z^*$}}
                	(10,0) node[anchor=north] {$NS_m$};
            \draw[->] (0,0) -- (0,11) node[anchor=east] {$r$};
            \draw	(0,0) node[anchor=east] {$0$}
				(0,1) node[anchor=east] {$\hat{p}_{N}^{m}$}
				(0,2) node[anchor=east] {$\hat{p}_{N-1}^{m}$}
				(0,3) node[anchor=east] {$\hat{p}_{N-2}^{m}$}
				(0,6) node[anchor=east] {\textbf{$r^*$}}
				(0,8) node[anchor=east] {$\hat{p}_{3}^{m}$}
				(0,9) node[anchor=east] {$\hat{p}_{2}^{m}$}
                (0,10) node[anchor=east] {$\hat{p}_{1}^{m}$};           
                
            \draw[step=1cm,gray,very thin] (0,0) grid (10,10);            

            \draw[line width=0.5mm,red] (0,10) -- (1,10);
            \draw[line width=0.5mm,red] (1,10) -- (1,9);
            \draw[line width=0.5mm,red] (1,9) -- (2,9);
            \draw[line width=0.5mm,red] (2,9) -- (2,8);
            \draw[line width=0.5mm,red] (2,8) -- (3,8);
            \draw[line width=0.5mm,red] (3,8) -- (3,7);
            \draw[line width=0.5mm,red] (5,7) -- (6,7);
            \draw[line width=0.5mm,red] (6,7) -- (6,4);
            \draw[line width=0.5mm,red] (7,3) -- (8,3);
            \draw[line width=0.5mm,red] (8,3) -- (8,2);
            \draw[line width=0.5mm,red] (8,2) -- (9,2);
            \draw[line width=0.5mm,red] (9,2) -- (9,1);
            \draw[line width=0.5mm,red] (9,1) -- (10,1);
            \draw[line width=0.5mm,red] (10,1) -- (10,0);

            \draw[line width=0.5mm,blue] (0,0) -- (10,10); 

            \draw [fill=red] (6,6) circle (0.1cm);

            \draw (3,9) node[red,rotate=-45,anchor=north] {\scriptsize{SP Reward Decision}};
            \draw  (3,3) node[blue,rotate=45,anchor=south] {\scriptsize{Users Reward Preference $Z(r)$}};
              
        \end{tikzpicture}	
	\caption{Choosing an optimal reward ($r$) for $N$ users with a similar and uniformly distributed mobility patterns over $L\rightarrow\infty$ locations in the centralized caching scheme}
	\label{Fig:Choosing_Reward_Similarity}
\end{figure}
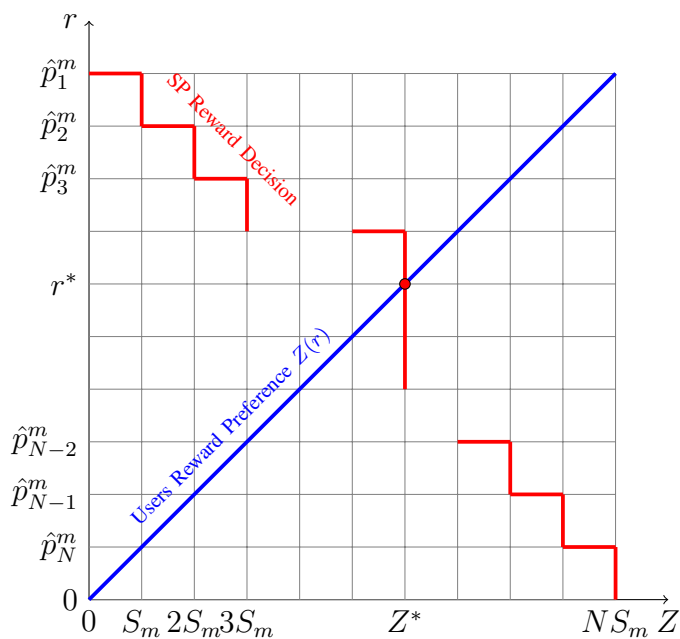

\section{Decentralized Caching Scheme}\label{Sec:Decentralized_Caching}

In the decentralized caching scheme, users make the caching decision based on the reward assigned by the SP. Users aim to leverage the lower network price at off-peak times to cache some data contents for their future request. They also share their proactive downloads with others to minimize their expected payments. We first define users payment function and then state the problem. We introduce an optimal decentralized caching policy and then compare it with the centralized caching policy mentioned in Section \ref{Sec:Centralized_Optimal}.

\subsection{Problem Statement}

In the flat pricing scenario, users behave reactively by requesting each data content $m$ at time $t$ with a probability $p_{n,t}^m$. Therefore, the time-averaged expected payment of user $n$ is given by:
	\begin{equation}\label{Eq:Payment_Reactive}
		\begin{split}
           	\mu_{n}^{\mathcal{R}} &= \limsup_{T\to\infty} \frac{1}{T} \sum_{t=1}^{T} \sum_{m=1}^{M} S_m p_{n,t}^{m} \\
           \end{split}
      \end{equation}
where the superscript $\mathcal{R}$ indicates the reactive operation. In the proposed model, we assume that users are aware of their demand and mobility profiles over $T$ time slots. Each user $n$ caches an amount $x_{n}^{m}$ of data item $m$ for a future possible request. Moreover, user $n$ transfers this data to other users through the D2D communication in any time slot $t$. Users replace the data stored in their devices when it is expired at the end of the day (i.e. at the end of time slot $T$). In particular, users cache data at the beginning of the day and share it throughout the rest of the day. We assume that users share their proactive downloads between them for free and they don't pay for getting their request from others. We assume that the peak price is normalized and the off-peak price is denoted by $r^{'}=1-r$, where $r$ is the reward received from the SP. In particular, instead of having $y_p,y_o$ for the peak and off-peak prices, as in \cite{Hosny2015Towards}, we have 1 and $r^{'}$ in this model. We still assume the maximum price constraint considered before with $\hat{y}=1$. Assuming that $y_p=\hat{y}$, it is enough to have an off-peak price including the SP reward. Hence, under this proactive model, the expected payment of user $n$ at time $t$ is given by:
    \begin{equation}\label{Eq:Payment_Proactive}
        \begin{split}
            \mu_{n,t}^{\mathcal{P}} & = \sum_{m=1}^{M} \sum_{k=2}^{N-1} \sum_{a^n_k \in \mA_n^k} \biggl(S_m - \sum_{j \in a_n^k} x_j^m \biggr)^+ p_{n,t}^{m} \underbrace{\sum_{l=1}^{L} \prod_{j \in a_n^k} \theta_{j,t}^{l} \prod_{i \notin a_n^k} \Bigl(1-\theta_{i,t}^{l}\Bigr)}_\text{meeting some users}\\
            & + \sum_{m=1}^{M} \biggl(S_m - \sum_{k=1}^{N} x_k^m \biggr)^+ p_{n,t}^{m} \underbrace{\sum_{l=1}^{L} \prod_{j=1}^N \theta_{j,t}^{l}}_\text{meeting all users} + \underbrace{r^{'} \sum_{m=1}^{M} x_n^m}_\text{caching cost}\\
            & + \sum_{m=1}^{M} \biggl( S_m - x_n^m \biggr) p_{n,t}^{m} \underbrace{\biggl( 1 - \sum_{l=1}^{L} \sum_{k=2}^{N} \sum_{a_n^k \in \mA_n^k} \prod_{j \in a_n^k} \theta_{j,t}^{l} \prod_{i \notin a_n^k} \Bigl(1-\theta_{i,t}^{l}\Bigr) \biggr)}_\text{user $n$ is alone}
        \end{split}
    \end{equation}
where the superscript $\mathcal{P}$ indicates the \emph{proactive} operation and $\mA_k^n$ is the set of all k-tuples indices including user $n$ . i.e. 
	\begin{equation*}
	    \mA_n^k = \Bigl\{a_k^n := \bigl(n,j_1,\cdots,j_{k-1}\bigr),j_i \neq n, \forall i \Bigr\},
	\end{equation*}
where $|\mA_k^n| = \binom{N-1}{k-1}$. The expected payment in (\ref{Eq:Payment_Proactive}) captures all the cases when user $n$ meets some users, when he meets all other users or when he is alone. He also pays $r^{'}$ for caching an amount $x_n^m$ of each content $m$. The time-averaged expected payment of user $n$ under the proposed model is given by:
    \begin{equation}\label{Eq:Payment_Proactive_Average}
        \mu_n^{\mathcal{P}} = \limsup_{T\to\infty} \frac{1}{T} \sum_{t=1}^{T} \mathbb{E} \biggl[ \mu_{n,t}^\mathcal{P} \biggr]
    \end{equation}


User's gain is the difference between the \emph{reactive} payment and the \emph{proactive} payment, under the proposed model, which is denoted by $\bigtriangleup \mu_n = \mu_n^\mathcal{R}-\mu_n^\mathcal{P}$. Users save some of their payment by finding the requested data items in their local cache or with others user in their neighborhood. User's objective is to achieve a positive gain (i.e. $\bigtriangleup \mu_n > 0$) by finding an optimal caching policy $\{x_n^{m*}\}_m$ which minimizes his time-averaged expected payment. The cached amount of data item $m$ at each user cannot exceed its size as mentioned in (\ref{Eq:Const}). Therefore, the problem is defined as
	\begin{equation}\label{Eq:User_Optimization_Problem}
		\begin{aligned}
			& \min \hspace{5mm} \mu_{n}^{\mathcal{P}}\\
                & \text{s.t.} \hspace{5mm} 0 \leq x_n^m \leq S_m, \hspace{5mm} \forall m.
		\end{aligned}
	\end{equation}

\subsection{Optimal Decentralized Caching Policy Analysis}\label{Sec:Optimal_Distributed}

In this section, we introduce an optimal decentralized caching policy which achieves a minimum payment for users. We can see from (\ref{Eq:Payment_Proactive}) that the objective function in (\ref{Eq:User_Optimization_Problem}) for user $n$ depends on the decision of the other users. Therefore, we start by the assumption that each user has a complete and perfect information about others and then discuss the sufficient statistics required to find his optimal decision. Moreover, without considering a memory constraint, we can decompose the problem in (\ref{Eq:User_Optimization_Problem}) to $M$ sub-problems and solve it for each content $m$, separately. We will introduce an optimal decentralized caching policy without considering any memory constraint. We discuss the effect of the memory constraint and how to choose an optimal memory size in Section \ref{Sec:Choosing_Memory}. To illustrate the idea of our analysis, we start by considering two simple cases for $N=2,3$ and then use them to generalize the solution.

\subsubsection{\textbf{\underline{Case Study ($N=2$)}}}

For simplicity, we start by $T=1$, and then extend it to any value of $T$. In this case, the suffix $t$ can be dropped and the expected payment of user $1$ will be
	\begin{equation} \label{Eq:Payment_N=2_T=1}
		\begin{aligned}
           	\mu_{1}^{\mathcal{P}} & = \sum_{m=1}^{M} \biggl(S_m - \Bigl(x_1^m+x_2^m\Bigr) \biggr)^+ p_1^m \sum_{l=1}^{L} \theta_1^l \theta_2^l\\
                & + \sum_{m=1}^{M} \biggl(S_m - x_1^m \biggr) p_1^m \biggl(1-\sum_{l=1}^{L} \theta_1^l \theta_2^l \biggr) + r^{'} \sum_{m=1}^{M} x_1^m
           \end{aligned}
	\end{equation}
Note that the optimal decision of user $1$ depends on the decision of user $2$. The problem decomposes to $M$ sub-problems and we have two sub-cases: either $x_1^m+x_2^m<S_m$ leading to a linear program (LP), where:
    \begin{equation} \label{Eq:Payment_N=2_T=1_Case1}
        \begin{aligned}
            \mu_1^{\mathcal{P}} & = \underbrace{\sum_{m=1}^{M} S_m p_1^m}_\text{reactive payment} + \underbrace{\sum_{m=1}^{M} \Bigl(r^{'}-p_1^m\Bigr) x_1^m}_\text{caching payment} - \underbrace{\sum_{m=1}^{M} x_2^m p_1^m \sum_{l=1}^{L} \theta_1^l \theta_2^l}_\text{sharing saving}
        \end{aligned}
    \end{equation}
and his optimization problem will be
	\begin{equation} \label{Eq:Problem_N=2_T=1_Case1}
		\begin{aligned}
			& \min \hspace{5mm} \sum_{m=1}^{M} \Bigl(r-p_1^m\Bigr) x_1^m\\
                & \text{s.t.} \hspace{8mm} 0 \leq x_1^m < S_m - x_2^m, \hspace{5mm} \forall m.
		\end{aligned}
	\end{equation}
or $x_1^m+x_2^m \geq S_m$ which leads to another LP, where:
    \begin{equation} \label{Eq:Payment_N=2_T=1_Case2}
        \begin{aligned}
            \mu_1^{\mathcal{P}} & = \underbrace{\sum_{m=1}^{M} S_m p_1^m}_\text{reactive payment} + \underbrace{\sum_{m=1}^{M} \Biggl(r^{'}-p_1^m \biggl(1-\sum\limits_{l=1}^{L} \theta_1^l \theta_2^l \biggr) \Biggr) x_1^m}_\text{caching payment} - \underbrace{\sum_{m=1}^{M} S_m p_1^m \sum_{l=1}^{L} \theta_1^l \theta_2^l}_\text{sharing saving}
        \end{aligned}
    \end{equation}
and his optimization problem will be
           \begin{equation} \label{Eq:Problem_N=2_T=1_Case2}
                \begin{aligned}
                & \min \hspace{5mm} \sum_{m=1}^{M} \Biggl(r^{'}-p_1^m \biggl(1-\sum\limits_{l=1}^{L} \theta_1^l \theta_2^l \biggr) \Biggr) x_1^m\\
                & \text{s.t.} \hspace{8mm} S_m - x_2^m \leq x_1^m \leq S_m, \hspace{5mm} \forall m.
                \end{aligned}
            \end{equation}  
The first term in \ref{Eq:Payment_N=2_T=1_Case1} and \ref{Eq:Payment_N=2_T=1_Case2} represents the reactive payment, the second term represents the payment corresponding to caching these data contents. The last term represents the saving in payment achieved by sharing the proactive download of user $2$. This saving gain depends on the meeting probability between user $1$ and $2$. We use the Best Response (BR) analysis to find the Sub-game Perfect Nash Equilibrium (SPNE) between them, where
    \begin{equation}
        \begin{aligned}
            \mathcal{B}_1 \Bigl(x_2^m\Bigr) &= \argmin_{0 \leq x_1^m \leq S_m} \mu_{1}^{\mathcal{P}} \Bigl( x_1^m, x_2^m \Bigr)\\
            \mathcal{B}_2 \Bigl(x_1^m\Bigr) &= \argmin_{0 \leq x_2^m \leq S_m} \mu_{2}^{\mathcal{P}} \Bigl( x_1^m, x_2^m \Bigr)\\
        \end{aligned}
    \end{equation}

User $1$ will consider all possible decisions of user $2$ and then make the decision that minimizes his payment for each case. Considering the two sub-cases when $x_1^m+x_2^m<S_m$ and $x_1^m+x_2^m \geq S_m$, we can plot the payment of user $1$ versus $x_2^m$ as shown in Figure \ref{Fig:Expected_Payment_Case1_N=2}. We can draw a similar curve for the payment of user $2$ as function of $x_1^m$. Based on these payment functions, we can come up with the best response shown in Figure \ref{Fig:Best_Response_N=2}. We can see that users best response depends on the comparison between the caching cost $r^{'}$ and their interest and mobility statistics. In particular, if $r^{'}<p_1^m \biggl(1-\sum\limits_{l=1}^{L} \theta_1^l \theta_2^l \biggr)$, user $1$ caches this content regardless of what user $2$ does, since its price is very low. If $r^{'}>p_1^m$, user $1$ will not have any incentive to cache this content, since its price is high. When $r^{'}$ lies between $p_1^m \biggl(1-\sum\limits_{l=1}^{L} \theta_1^l \theta_2^l \biggr)$ and $p_1^m$, user $1$ prefers to share the payment with user $2$, i.e. if user $2$ caches an amount $x$ from this content, user $1$ opts to cache an amount $S_m-x$. In particular, when $r^{'}$ lies in this region, partial caching is an optimal solution.

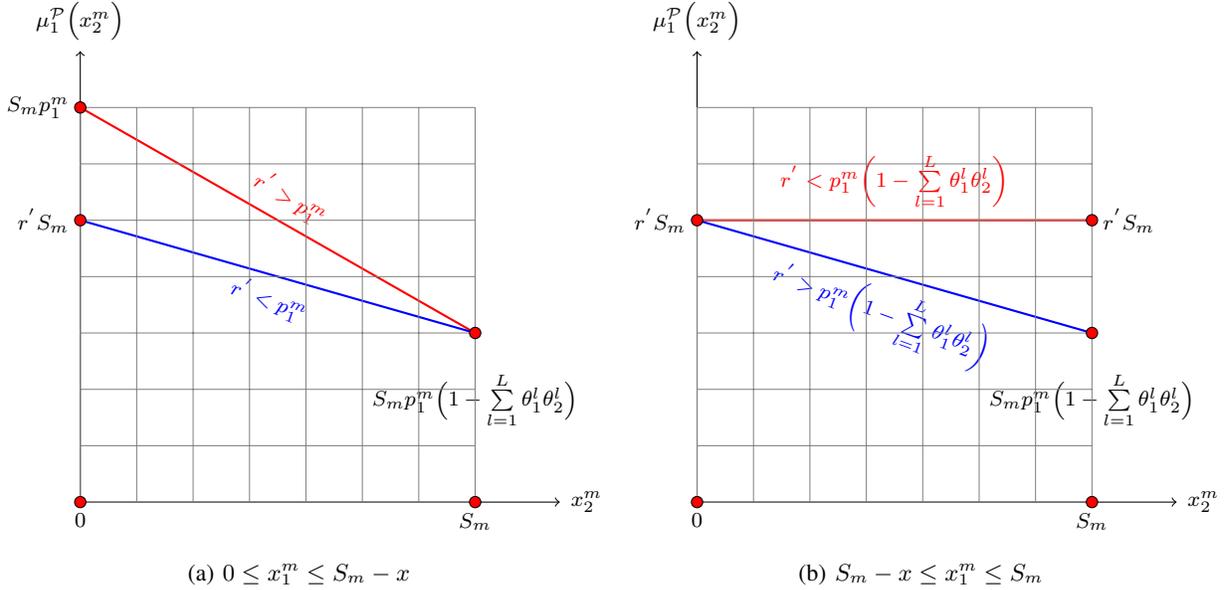
\begin{figure}[h!]
    \centering
    \subfloat[$0 \leq x_1^m \leq S_m - x$]{
        \centering
        \begin{tikzpicture}[scale=0.75]
            \draw[->] (0,0) -- (8.5,0) node[anchor=west] {\scriptsize{$x_2^m$}};
            \draw	(0,0) node[anchor=north] {\scriptsize{$0$}}
            		(7,0) node[anchor=north] {\scriptsize{$S_m$}};
            \draw[->] (0,0) -- (0,8) node[anchor=south] {\scriptsize{$\mu_{1}^\mathcal{P} \Bigl(x_2^m\Bigr)$}};
            \draw	(0,0) node[anchor=east] {}
            		(0,5) node[anchor=east] {\scriptsize{$r^{'} S_m$}}
            		(0,7) node[anchor=east] {\scriptsize{$S_m p_1^m$}}
            		(7,2.5) node[anchor=north] {\scriptsize{$S_m p_1^m \Bigl( 1 - \sum\limits_{l=1}^{L} \theta_1^l \theta_2^l \Bigr)$}}
            		(3.5,5) node[red,rotate=-30,anchor=south] {\scriptsize{$r^{'} > p_1^m$}}
            		(3.5,4) node[blue,rotate=-20,anchor=north] {\scriptsize{$r^{'} < p_1^m$}};  
            \draw[red,thick] (0,7) -- (7,3);
            \draw[blue,thick] (0,5) -- (7,3);
            \draw[step=1cm,gray,very thin] (0,0) grid (7,7);            
            \draw [fill=red] (0,0) circle (0.1cm);
            \draw [fill=red] (7,0) circle (0.1cm);
            \draw [fill=red] (7,3) circle (0.1cm);
            \draw [fill=red] (0,5) circle (0.1cm);
            \draw [fill=red] (0,7) circle (0.1cm);
            
        \end{tikzpicture}}
    \subfloat[$S_m - x \leq x_1^m \leq S_m$]{
        \centering
        \begin{tikzpicture}[scale=0.75]
            \draw[->] (0,0) -- (8.5,0) node[anchor=west] {\scriptsize{$x_2^m$}};
            \draw	(0,0) node[anchor=north] {\scriptsize{$0$}}
            		(7,0) node[anchor=north] {\scriptsize{$S_m$}};
            \draw[->] (0,0) -- (0,8) node[anchor=south] {\scriptsize{$\mu_{1}^\mathcal{P} \Bigl(x_2^m\Bigr)$}};
            \draw	(0,0) node[anchor=east] {}
            		(0,5) node[anchor=east] {\scriptsize{$r^{'} S_m$}}
            		(7,5) node[anchor=west] {\scriptsize{$r^{'} S_m$}}
            		(7,2.5) node[anchor=north] {\scriptsize{$S_m p_1^m \Bigl( 1 - \sum\limits_{l=1}^{L} \theta_1^l \theta_2^l \Bigr)$}}
            		(3.5,5) node[red,anchor=south] {\scriptsize{$r^{'} < p_1^m \biggl(1-\sum\limits_{l=1}^{L} \theta_1^l \theta_2^l \biggr)$}}
            		(3.5,4) node[blue,rotate=-20,anchor=north] {\scriptsize{$r^{'} > p_1^m \biggl(1-\sum\limits_{l=1}^{L} \theta_1^l \theta_2^l \biggr)$}};  
            \draw[red,thick] (0,5) -- (7,5);
            \draw[blue,thick] (0,5) -- (7,3);
            \draw[step=1cm,gray,very thin] (0,0) grid (7,7);            
            \draw [fill=red] (0,0) circle (0.1cm);
            \draw [fill=red] (7,0) circle (0.1cm);
            \draw [fill=red] (7,3) circle (0.1cm);
            \draw [fill=red] (0,5) circle (0.1cm);
            \draw [fill=red] (7,5) circle (0.1cm);

        \end{tikzpicture}}
    \caption{Expected payment for user $1$}
    \label{Fig:Expected_Payment_Case1_N=2}
\end{figure}

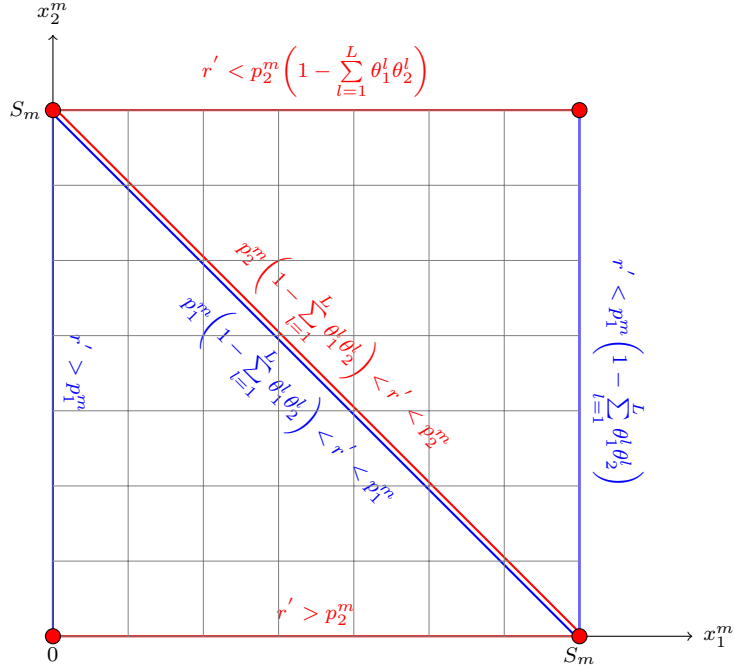
\begin{figure}[h!]
    \centering
    \begin{tikzpicture}[scale=1.0]
        \draw[->] (0,0) -- (8.5,0) node[anchor=west] {\scriptsize{$x_1^m$}};
        \draw	(0,0) node[anchor=north] {\scriptsize{$0$}}
        		(7,0) node[anchor=north] {\scriptsize{$S_m$}};
        \draw[->] (0,0) -- (0,8) node[anchor=south] {\scriptsize{$x_2^m$}};
        \draw	(0,0) node[anchor=east] {}
        		(0,7) node[anchor=east] {\scriptsize{$S_m$}}
        		(3.5,7) node[red,anchor=south] {\scriptsize{$r^{'} < p_2^m \biggl(1-\sum\limits_{l=1}^{L} \theta_1^l \theta_2^l \biggr)$}}
        		(3.5,3.5) node[red,rotate=-45,anchor=south] {\scriptsize{$p_2^m \biggl(1-\sum\limits_{l=1}^{L} \theta_1^l \theta_2^l \biggr) < r^{'} < p_2^m$}}
        		(3.5,0) node[red,anchor=south] {\scriptsize{$r^{'} > p_2^m$}}
        		(7,3.5) node[blue,rotate=-90,anchor=south] {\scriptsize{$r^{'} < p_1^m \biggl(1-\sum\limits_{l=1}^{L} \theta_1^l \theta_2^l \biggr)$}}
                (3.5,3.5) node[blue,rotate=-45,anchor=north] {\scriptsize{$p_1^m \biggl(1-\sum\limits_{l=1}^{L} \theta_1^l \theta_2^l \biggr) < r^{'} < p_1^m$}}
                (0,3.5) node[blue,rotate=-90,anchor=south] {\scriptsize{$r^{'} > p_1^m$}};  
        \draw[red,thick] (0,0) -- (7,0);
        \draw[red,thick] (7.05,0) -- (0,7.05);
        \draw[red,thick] (0,7) -- (7,7);
        \draw[blue,thick] (0,0) -- (0,7);
        \draw[blue,thick] (0,6.95) -- (6.95,0);
        \draw[blue,thick] (7,0) -- (7,7);
        \draw[step=1cm,gray,very thin] (0,0) grid (7,7);            
        \draw [fill=red] (0,0) circle (0.1cm);
        \draw [fill=red] (7,0) circle (0.1cm);
        \draw [fill=red] (0,7) circle (0.1cm);
        \draw [fill=red] (7,7) circle (0.1cm);
    \end{tikzpicture}
    \caption{Best Response (BR) of user $1$}
    \label{Fig:Best_Response_N=2}
\end{figure}

Now, without loss of generality, we can assume that $p_1^m>p_2^m$ and hence we have $p_1^m \biggl(1-\sum\limits_{l=1}^{L} \theta_1^l \theta_2^l \biggr)>p_2^m \biggl(1-\sum\limits_{l=1}^{L} \theta_1^l \theta_2^l \biggr)$.  We can consider two sub-cases, either $p_2^m < p_1^m \biggl(1-\sum\limits_{l=1}^{L} \theta_1^l \theta_2^l \biggr)$ or $p_2^m \geq p_1^m \biggl(1-\sum\limits_{l=1}^{L} \theta_1^l \theta_2^l \biggr)$, which leads to the solutions shown in Figures \ref{Fig:Opt_Soln_Char_N=2_T=1_Case_1} and \ref{Fig:Opt_Soln_Char_N=2_T=1_Case_2}, respectively. In Figure \ref{Fig:Opt_Soln_Char_N=2_T=1_Case_1}, we can see that when $r^{'}<p_2^m \biggl(1-\sum\limits_{l=1}^{L} \theta_1^l \theta_2^l \biggr)$, both users will have enough incentive to cache the content, since its price is very low. When $r^{'}>p_1^m$, both users will opt not to cache, since the price is high. When $p_2^m \biggl(1-\sum\limits_{l=1}^{L} \theta_1^l \theta_2^l \biggr) \leq r^{'} < p_1^m \biggl(1-\sum\limits_{l=1}^{L} \theta_1^l \theta_2^l \biggr)$, user $1$ caches this content and user $2$ takes it from him. When $p_1^m \biggl(1-\sum\limits_{l=1}^{L} \theta_1^l \theta_2^l \biggr) \leq r^{'} < p_1^m$, user $1$ prefers to share the payment with user $2$. But since $r^{'}>p_2^m$, user $2$ will not have any incentive to participate in caching this content. Therefore, user $1$ will cache the whole content alone. This sub-case does not have any ambiguity and there exits a unique SPNE between both users.

    \begin{figure}[h!]
        \centering
        \begin{tikzpicture}[scale=0.75]
            \draw[->,thick] (0,0) -- (16,0) node[anchor=west] {$r^{'}$};
            \draw	(0,0) node[anchor=north] {\scriptsize{$0$}}
        		    (3,0) node[anchor=north] {\scriptsize{$p_2^m \biggl(1-\sum\limits_{l=1}^{L} \theta_1^l \theta_2^l \biggr)$}}
        		    (6,0) node[anchor=north] {\scriptsize{$p_2^m$}}
        		    (9,0) node[anchor=north] {\scriptsize{$p_1^m \biggl(1-\sum\limits_{l=1}^{L} \theta_1^l \theta_2^l \biggr)$}}
        		    (12,0) node[anchor=north] {\scriptsize{$p_1^m$}}
        		    (15,0) node[anchor=north] {\scriptsize{$1$}};
            \draw	(1.5,1) node[anchor=north] {$(S_m,S_m)$}
        		    (4.5,1) node[anchor=north] {$(S_m,0)$}
        		    (7.5,1) node[anchor=north] {$(S_m,0)$}
        		    (10.5,1) node[anchor=north] {$(S_m,0)$}
        		    (13.5,1) node[anchor=north] {$(0,0)$};
            \draw [fill=red] (0,0) circle (0.1cm);
            \draw [fill=red] (3,0) circle (0.1cm);
            \draw [fill=red] (6,0) circle (0.1cm);
            \draw [fill=red] (9,0) circle (0.1cm);
            \draw [fill=red] (12,0) circle (0.1cm);
            \draw [fill=red] (15,0) circle (0.1cm);
            
        \end{tikzpicture}
        \caption{Decentralized caching policy for $N=2,T=1$ when $p_2^m < p_1^m \biggl(1-\sum\limits_{l=1}^{L} \theta_1^l \theta_2^l \biggr)$}
        \label{Fig:Opt_Soln_Char_N=2_T=1_Case_1}
    \end{figure}
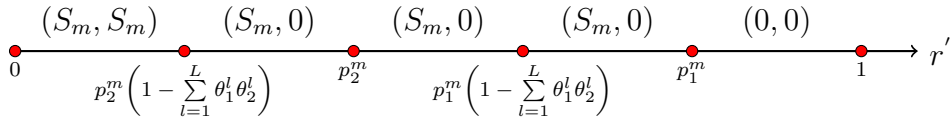

In Figure \ref{Fig:Opt_Soln_Char_N=2_T=1_Case_2}, when $r^{'}$ is very small, such that it is less than $p_2^m \biggl(1-\sum\limits_{l=1}^{L} \theta_1^l \theta_2^l \biggr)$ and $p_1^m \biggl(1-\sum\limits_{l=1}^{L} \theta_1^l \theta_2^l \biggr)$, both users cache this content. When $p_2^m \biggl(1-\sum\limits_{l=1}^{L} \theta_1^l \theta_2^l \biggr) \leq r^{'} < p_1^m \biggl(1-\sum\limits_{l=1}^{L} \theta_1^l \theta_2^l \biggr)$, user $1$ still has an incentive to cache this content, and hence user $2$ will depend on him and opt not to cache. When $p_1^m \biggl(1-\sum\limits_{l=1}^{L} \theta_1^l \theta_2^l \biggr) \leq r^{'} < p_2^m$, partial caching will be an optimal solution. So if user $2$ caches an amount $x$, user $1$ completes it by caching $S_m-x$. Actually, any value $0\leq x \leq S_m$ leads to a Nash equilibrium. This means that we have a non-unique Nash equilibrium in this region. Therefore, it is important to find another dynamic to choose one of these equilibria, as discussed in Section \ref{Sec:Fair_Distributed}.
    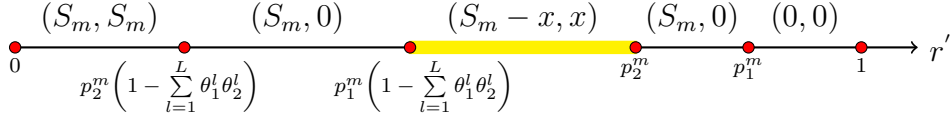
\begin{figure}[h!]
        \centering
        \begin{tikzpicture}[scale=0.75]
            \draw[->,thick] (0,0) -- (16,0) node[anchor=west] {$r^{'}$};
            \draw	(0,0) node[anchor=north] {\scriptsize{$0$}}
        		    (2.75,0) node[anchor=north] {\scriptsize{$p_2^m \biggl(1-\sum\limits_{l=1}^{L} \theta_1^l \theta_2^l \biggr)$}}
        		    (7.25,0) node[anchor=north] {\scriptsize{$p_1^m \biggl(1-\sum\limits_{l=1}^{L} \theta_1^l \theta_2^l \biggr)$}}
        		    (11,0) node[anchor=north] {\scriptsize{$p_2^m$}}
        		    (13,0) node[anchor=north] {\scriptsize{$p_1^m$}}
        		    (15,0) node[anchor=north] {\scriptsize{$1$}};
            \draw	(1.5,1) node[anchor=north] {$(S_m,S_m)$}
        		    (5,1) node[anchor=north] {$(S_m,0)$}
        		    (9,1) node[anchor=north] {$(S_m-x,x)$}
        		    (12,1) node[anchor=north] {$(S_m,0)$}
        		    (14,1) node[anchor=north] {$(0,0)$};
            \draw [fill=yellow,yellow,thick] (7,-0.1) rectangle (11,0.1);
            \draw [fill=red] (0,0) circle (0.1cm);
            \draw [fill=red] (3,0) circle (0.1cm);
            \draw [fill=red] (7,0) circle (0.1cm);
            \draw [fill=red] (11,0) circle (0.1cm);
            \draw [fill=red] (13,0) circle (0.1cm);
            \draw [fill=red] (15,0) circle (0.1cm);
            
        \end{tikzpicture}
	    \caption{Decentralized caching policy for $N=2,T=1$ when $p_2^m \geq p_1^m \biggl(1-\sum\limits_{l=1}^{L} \theta_1^l \theta_2^l \biggr)$}
        \label{Fig:Opt_Soln_Char_N=2_T=1_Case_2}
    \end{figure}

\subsubsection{\textbf{\underline{Case Study ($N=3$)}}}

We start by $T=1$, and then we can extend the result for any value of $T$. The suffix $t$ can be dropped and the expected payment of user $1$ can be written as:
    \begin{equation} \label{Eq:Payment_N=3_T=1}
        \begin{aligned}
            \mu_{1}^{\mathcal{P}} & = \sum_{m=1}^{M} \biggl(S_m - \Bigl(x_1^m+x_2^m\Bigr) \biggr)^+ p_1^m \sum_{l=1}^{L} \theta_1^l \theta_2^l \Bigl(1-\theta_3^l\Bigr)\\
            & + \sum_{m=1}^{M} \biggl(S_m - \Bigl(x_1^m+x_3^m\Bigr) \biggr)^+ p_1^m \sum_{l=1}^{L} \theta_1^l \theta_3^l \Bigl(1-\theta_2^l\Bigr) \\
            & + \sum_{m=1}^{M} \biggl(S_m - \Bigl(x_1^m+x_2^m+x_3^m\Bigr) \biggr)^+ p_1^m \sum_{l=1}^{L} \theta_1^l \theta_2^l \theta_3^l + r^{'} \sum_{m=1}^{M} x_1^m \\
            & + \sum_{m=1}^{M} \biggl(S_m - x_1^m \biggr) p_1^m \biggl(1-\underbrace{\sum_{l=1}^{L} \Bigl(\theta_1^l \theta_2^l+ \theta_1^l \theta_3^l - \theta_1^l \theta_2^l \theta_3^l \Bigr)}_\text{$v1$} \biggr)
        \end{aligned}
    \end{equation}

Following the same best response analysis, user $1$ determines his best response based on the decision of users $2$ and $3$. Therefore, $\mathcal{B}_1 (x_2^m,x_3^m)$ is the caching decision which achieves minimum payment for the corresponding values of $x_2^m$ and $x_3^m$. The optimal solution of user $1$ is shown in Figure \ref{Fig:Opt_Soln_Char_N=3_T=1}. Basically, each user $n$ decides whether he will be caching the content alone, sharing the payment with others, or discarding it at all based on the relation between $r^{'}$, $p_n^m$ and $p_n^m\bigl(1-v_n\bigr)$.
    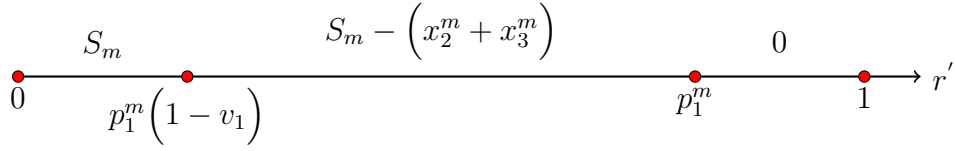
\begin{figure}[h!]
        \centering
        \begin{tikzpicture}[scale=0.75]
            \draw[->,thick] (0,0) -- (16,0) node[anchor=west] {$r^{'}$};
            \draw	(0,0) node[anchor=north] {$0$}
        		    (3,0) node[anchor=north] {$p_1^m \Bigl(1 - v_1 \Bigr)$}
        		    (12,0) node[anchor=north] {$p_1^m$}
        		    (15,0) node[anchor=north] {$1$};
            \draw	(1.5,1) node[anchor=north] {$S_m$}
        		    (7.5,1.5) node[anchor=north] {$S_m - \Bigl( x_2^m + x_3^m \Bigr)$}
        		    (13.5,1) node[anchor=north] {$0$};
            \draw [fill=red] (0,0) circle (0.1cm);
            \draw [fill=red] (3,0) circle (0.1cm);
            \draw [fill=red] (12,0) circle (0.1cm);
            \draw [fill=red] (15,0) circle (0.1cm);
            
        \end{tikzpicture}
	    \caption{Optimal solution of user $1$ for $N=3$ in the decentralized caching scheme}
        \label{Fig:Opt_Soln_Char_N=3_T=1}
    \end{figure}

The optimal solution of all users depends on the relation between their interest. For example, suppose $p_1^m>p_2^m>p_3^m$ and $p_1^m \Bigl(1-v_1\bigr)>p_2^m \Bigl(1-v_2\bigr)>p_3^m \Bigl(1-v_3\bigr)$. The optimal solution will be as shown in Figure \ref{Fig:Opt_Soln_Char_N=3_T=1_Ex_1}. When $r^{'}$ is small enough, all users cache the content. There are some other regions of $r^{'}$ where partial caching is an optimal solution. We notice that the non-unique equilibrium region expanded because partial caching may occur between user $1$ and $2$, $2$ and $3$ or $1,2$ and $3$. We emphasis here that the optimal solution depends on the relation between $p_1^m,p_2^m$ and $p_3^m$ and the relation between $v_1,v_2$ and $v_3$. The solution shown in Figure \ref{Fig:Opt_Soln_Char_N=3_T=1_Ex_1} considers one example but there are some other cases. However, the same idea applies to find the optimal solution in each regime of $r^{'}$.

    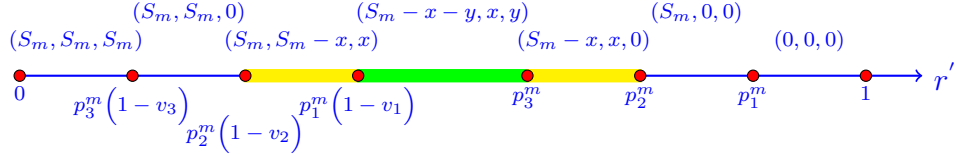
\begin{figure}[h!]
        \centering
        \begin{tikzpicture}[scale=0.75]
            \draw[->,blue,thick] (0,0) -- (16,0) node[anchor=west] {$r^{'}$};
            \draw	(0,0) node[anchor=north,blue] {\scriptsize{$0$}}
        		    (2,0) node[anchor=north,blue] {\scriptsize{$p_3^m \Bigl(1-v_3 \Bigr)$}}
        		    (4,-0.5) node[anchor=north,blue] {\scriptsize{$p_2^m \Bigl(1-v_2 \Bigr)$}}
        		    (6,0) node[anchor=north,blue] {\scriptsize{$p_1^m \Bigl(1-v_1 \Bigr)$}}
        		    (9,0) node[anchor=north,blue] {\scriptsize{$p_3^m$}}
        		    (11,0) node[anchor=north,blue] {\scriptsize{$p_2^m$}}
        		    (13,0) node[anchor=north,blue] {\scriptsize{$p_1^m$}}
        		    (15,0) node[anchor=north,blue] {\scriptsize{$1$}};
            \draw	(1,1) node[anchor=north,blue] {\scriptsize{$(S_m,S_m,S_m)$}}
        		    (3,1.5) node[anchor=north,blue] {\scriptsize{$(S_m,S_m,0)$}}
        		    (5,1) node[anchor=north,blue] {\scriptsize{$(S_m,S_m-x,x)$}}
        		    (7.5,1.5) node[anchor=north,blue] {\scriptsize{$(S_m-x-y,x,y)$}}
        		    (10,1) node[anchor=north,blue] {\scriptsize{$(S_m-x,x,0)$}}
        		    (12,1.5) node[anchor=north,blue] {\scriptsize{$(S_m,0,0)$}}
        		    (14,1) node[anchor=north,blue] {\scriptsize{$(0,0,0)$}};
            \draw [fill=yellow,yellow,thick] (4,-0.1) rectangle (6,0.1);
            \draw [fill=green,green,thick] (6,-0.1) rectangle (9,0.1);
            \draw [fill=yellow,yellow,thick] (9,-0.1) rectangle (11,0.1);
            \draw [fill=red] (0,0) circle (0.1cm);
            \draw [fill=red] (2,0) circle (0.1cm);
            \draw [fill=red] (4,0) circle (0.1cm);
            \draw [fill=red] (6,0) circle (0.1cm);
            \draw [fill=red] (9,0) circle (0.1cm);
            \draw [fill=red] (11,0) circle (0.1cm);
            \draw [fill=red] (13,0) circle (0.1cm);
            \draw [fill=red] (15,0) circle (0.1cm);

        \end{tikzpicture}
	    \caption{An example decentralized caching policy for $N=3$}
        \label{Fig:Opt_Soln_Char_N=3_T=1_Ex_1}
    \end{figure}

\subsubsection{\textbf{\underline{Optimal Policy for $N$-users}}}

Now, from the previous cases, we can infer the optimal decentralized caching policy for a general number of users as shown in Figure \ref{fig:Opt_Soln_Char_N_T=1}, where $\hat{p}_{n}^m = \frac{1}{T} \sum_{t=1}^{T} p_{n,t}^m$, $\tilde{p}_{n}^m = \frac{1}{T} \sum_{t=1}^{T} p_{n,t}^m \Bigl( 1 - v_{n,t} \Bigr)$ and $v_n$ is as defined in (\ref{Eq:Ranking_Optimal_N_N-1}), $\forall n\in\mathcal{N}$.  Each user compares the caching cost $r^{'}$ with his interest and mobility statistics $\hat{p}_{n}^m$ and $\tilde{p}_{n}^m$ to determine whether he is caching the whole content, sharing the cost with others, or discarding it at all. Since there are non-unique equilibrium for the partial caching regime, we are not able to show uniqueness of the SPNE. The following theorem states the existence of the SPNE.
    \begin{theorem}
    For a game of $N$ users, there exists a Subgame Perfect Nash Equilibrium (SPNE) between users.
    \end{theorem}
    \begin{proof}
		The existence follows from Debreu, Glicksberg and Fan (DGF) theorem since:
                \begin{itemize}
                    \item $x_n^m \in [ 0, S_m ], \forall n \in \mathcal{N}$ are compact and convex.
                    \item $\mu_n^\mathcal{P}$ are continuous over $[ 0, S_m ], \forall n \in \mathcal{N}$.
                    \item $\mu_n^\mathcal{P}$ are concave by its linearity in $x_1^m, x_2^m,\cdots,x_n^m$ (for each sub-case separately).
                \end{itemize}
	Optimality of the solution was shown by the best response analysis discussed before. 
    \end{proof}
    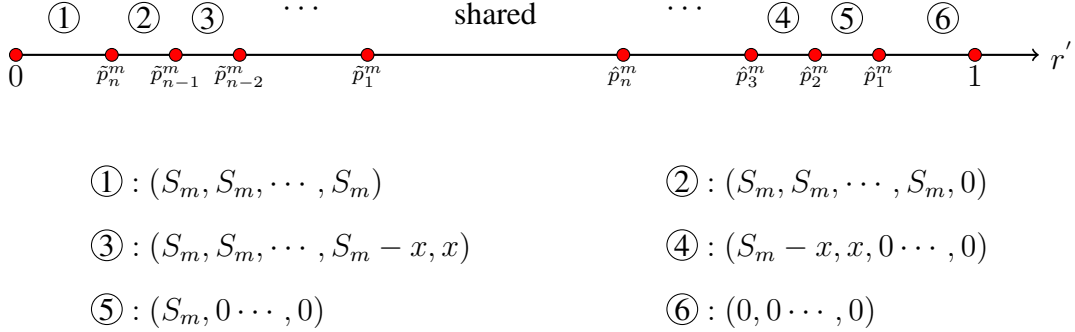
\begin{figure}[h!]
        \centering
        \begin{tikzpicture}[scale=0.85]
            \draw[->,thick] (0,0) -- (16,0) node[anchor=west] {$r^{'}$};
            \draw	(0,0) node[anchor=north] {$0$}
        		    (1.5,0) node[anchor=north] {\scriptsize{$\tilde{p}_n^m$}}
        		    (2.5,0) node[anchor=north] {\scriptsize{$\tilde{p}_{n-1}^m$}}
        		    (3.5,0) node[anchor=north] {\scriptsize{$\tilde{p}_{n-2}^m$}}                		    
    				(5.5,0) node[anchor=north] {\scriptsize{$\tilde{p}_1^m$}}
        		    (9.5,0) node[anchor=north] {\scriptsize{$\hat{p}_n^m$}}
        		    (11.5,0) node[anchor=north] {\scriptsize{$\hat{p}_3^m$}}
        		    (12.5,0) node[anchor=north] {\scriptsize{$\hat{p}_2^m$}}
        		    (13.5,0) node[anchor=north] {\scriptsize{$\hat{p}_1^m$}}
        		    (15,0) node[anchor=north] {$1$};
            \draw	(0.75,1) node[anchor=north] {\circled{1}}
        		    (2,1) node[anchor=north] {\circled{2}}
        		    (3,1) node[anchor=north] {\circled{3}}
        		    (4.5,1) node[anchor=north] {$\cdots$}
        		    (7.5,1) node[anchor=north] {shared}
        		    (10.5,1) node[anchor=north] {$\cdots$}
        		    (12,1) node[anchor=north] {\circled{4}}
        		    (13,1) node[anchor=north] {\circled{5}}
        		    (14.5,1) node[anchor=north] {\circled{6}};
            \draw [fill=red] (0,0) circle (0.1cm);
            \draw [fill=red] (1.5,0) circle (0.1cm);
            \draw [fill=red] (2.5,0) circle (0.1cm);
            \draw [fill=red] (3.5,0) circle (0.1cm);
            \draw [fill=red] (5.5,0) circle (0.1cm);
            \draw [fill=red] (9.5,0) circle (0.1cm);
            \draw [fill=red] (11.5,0) circle (0.1cm);
            \draw [fill=red] (12.5,0) circle (0.1cm);
            \draw [fill=red] (13.5,0) circle (0.1cm);
            \draw [fill=red] (15,0) circle (0.1cm);
    
    		\draw  (1,-2) node[anchor=west] {$\circled{1}: (S_m,S_m,\cdots,S_m)$}
    					(10,-2) node[anchor=west] {$\circled{2}: (S_m,S_m,\cdots,S_m,0)$}
    				   (1,-3) node[anchor=west] { $\circled{3}: (S_m,S_m,\cdots,S_m-x,x)$}
    				   (10,-3) node[anchor=west] {$\circled{4}: (S_m-x,x,0\cdots,0)$}
    					(1,-4) node[anchor=west] { $\circled{5}: (S_m,0\cdots,0)$}
    					(10,-4) node[anchor=west] {$\circled{6}: (0,0\cdots,0)$};
    					
        \end{tikzpicture}
    	\caption{Optimal decentralized caching policy for $N$ users}
        \label{fig:Opt_Soln_Char_N_T=1}
    \end{figure}

\subsection{Fair Caching Allocation}\label{Sec:Fair_Distributed}

When the caching cost $r^{'}$ lies in the regime where partial caching is an optimal solution, there exits a non-unique Nash equilibrium. Another dynamic need to be added to the game that allows the users to choose one of these equilibria \cite{harsanyi1988general,fevrier2006equilibrium}. A Nash equilibrium is considered \emph{payoff dominant} if it is Pareto superior to all other Nash equilibria in the game. Unfortunately, it is not clear if any of these equilibira has this feature. For example in the case of $N=2$, we can see that if user $2$ caches an amount $x$ of content $m$ and user $1$ completes it by caching an amount  $S_m-x$, the payment of user $1$, corresponding to this content, will be
    \begin{equation*}
            \mu_1^\mathcal{P} = r^{'} S_m - x \left[r-p_1^m\left(1-\sum\limits_{l=1}^{L}\theta_1^l \theta_2^l\right)\right]
    \end{equation*}
which is a decreasing function in $x$. In particular, any increase in $x$ is preferable to user $1$. On the contrary, the payment of user $2$, corresponding to this content, will be
    \begin{equation*}
            \mu_2^\mathcal{P} = S_m p_2^m \left(1-\sum\limits_{l=1}^{L}\theta_1^l \theta_2^l\right) + x \left[r^{'}-p_1^m\left(1-\sum\limits_{l=1}^{L}\theta_1^l \theta_2^l\right)\right]
    \end{equation*}
which is an increasing function in $x$. Therefore, user $2$ will try to reduce $x$ as much as possible. This means that the tension between both users will not lead them to a payoff dominant NE.

A Nash equilibrium is considered \emph{risk dominant} if it has the largest basin of attraction (i.e. is less risky). In particular, the more uncertainty players have about the actions of the other player(s), the more likely they will choose the strategy corresponding to it. Each user evaluates the risk corresponding to each NE, given that he doesn't know the reaction of the other users, and chooses the one with the least risk value (e.g. smallest expected payment). Unfortunately, this approach does not necessarily lead us to one of the Nash equilibria. Figure \ref{fig:Risk_Dominance} depicts the result obtained for an example of $N=2, T=1$, where $p_1^m=0.8,p_2^m=0.6$ and their meeting probability is $0.5$. We can see that there are some regimes of $r^{'}$ where $x_1^m+x_2^m$ exceeds $S_m$. In particular, for $0.4\leq r^{'} \leq 0.5$, the risk dominance solution does not lead to a Nash equilibrium. The corresponding payments are shown in Figure \ref{fig:Risk_Dominance} (b). User $1$ pays more than user $2$ since he is caching more. This also means that user $1$ is more affected by the risk dominance policy.

\begin{figure}
    \centering
    \begin{tabular}{cc}
        \subfloat[Users caching decisions]{\includegraphics[width = 0.50\textwidth]{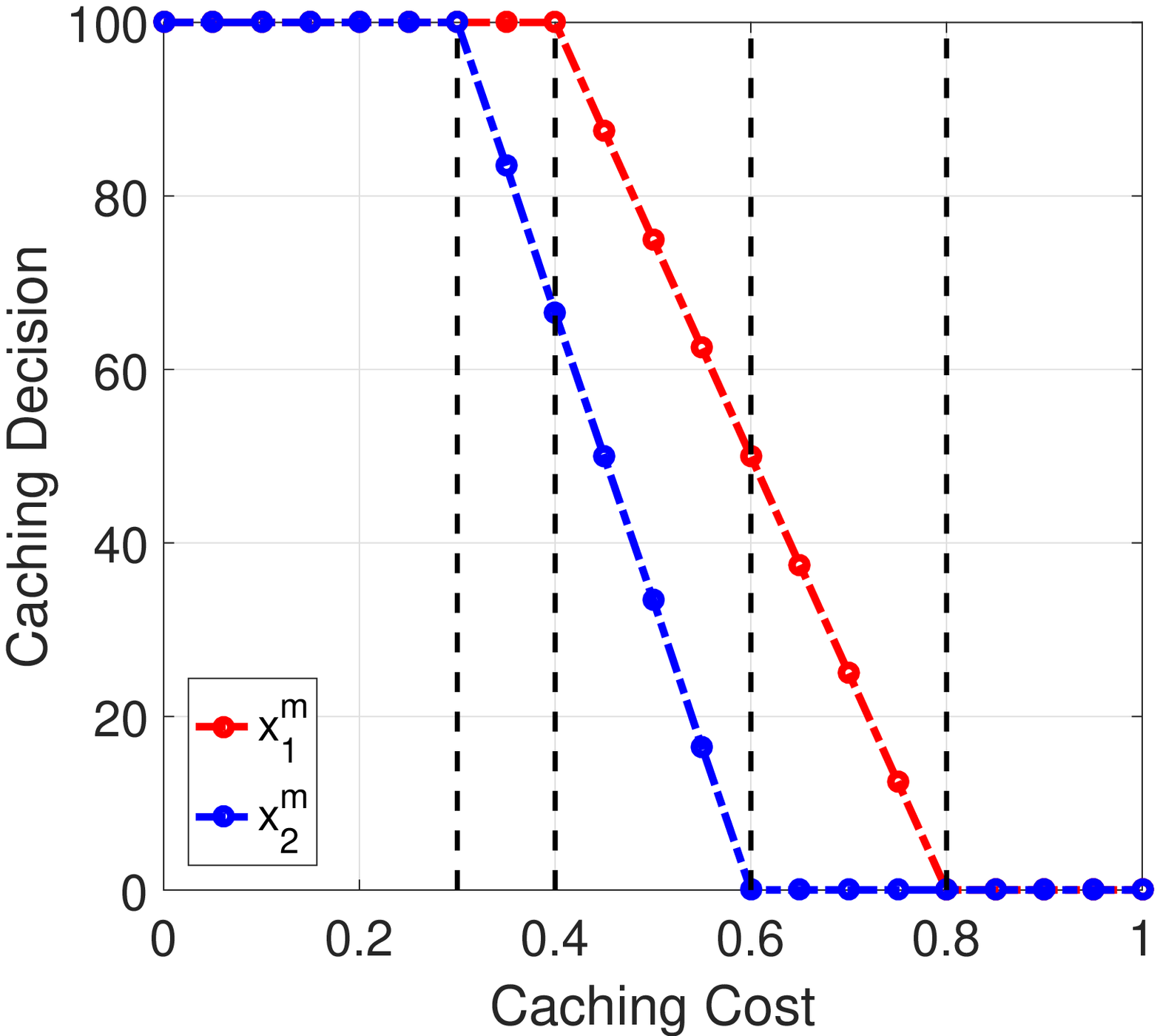}} &
        \subfloat[Expected user payment Comparison]{\includegraphics[width = 0.50\textwidth]{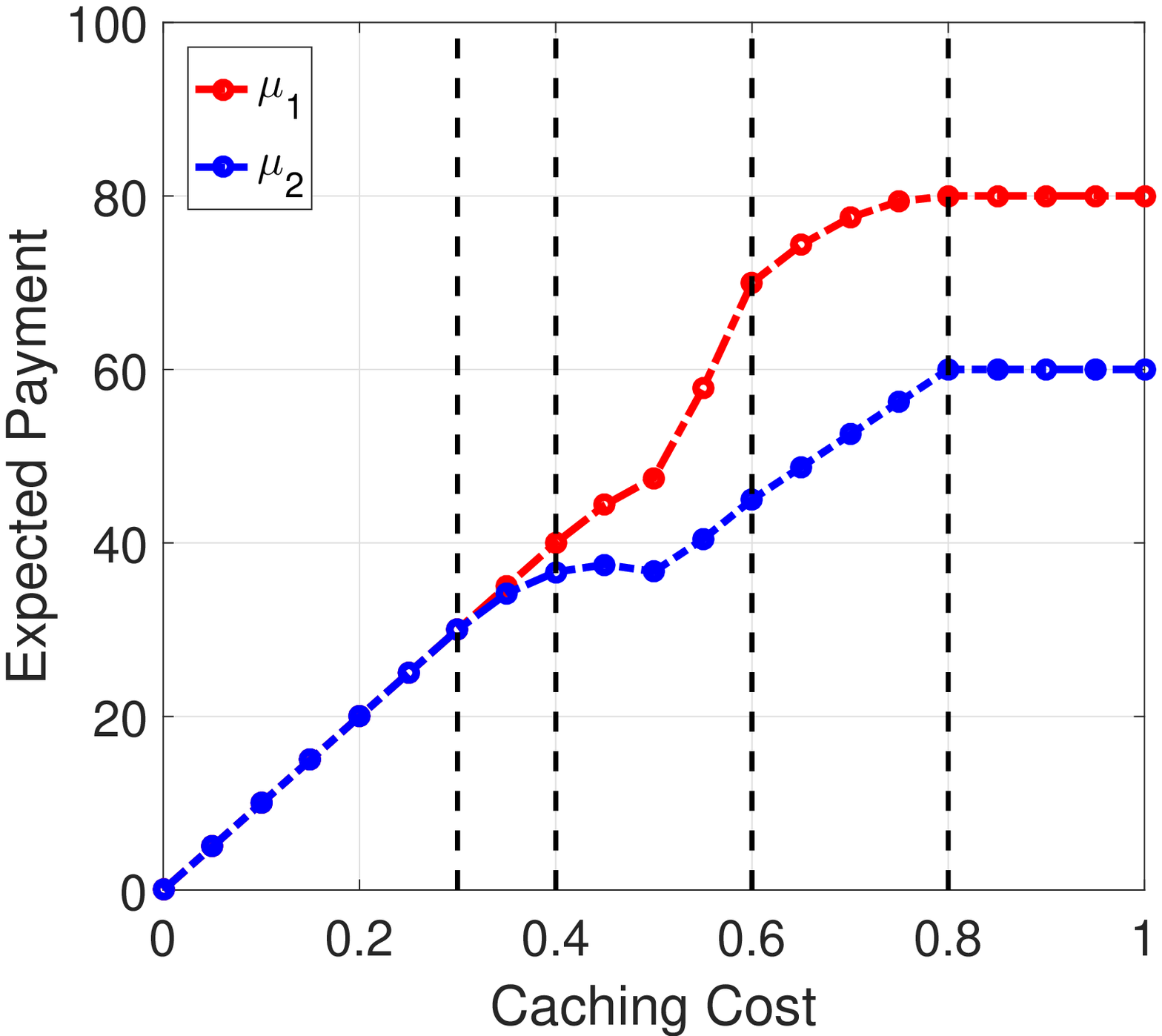}}
    \end{tabular}
    \caption{Risk dominance example for $N=2,T=1$ in the decentralized caching scheme}
    \label{fig:Risk_Dominance}
\end{figure}

Preplay communication is another way to coordinate between users. Users agree before playing the game on a certain strategy when $r^{'}$ lies in the partial caching regime. For example, they may agree on caching amounts proportional to their interests. This coordination may also be imposed by the SP who sets this rule for all users before playing the game. We know that users will pick one of the Nash equilibria since it allows them to minimize their payment. We adopt a fair allocation strategy for this case which is defined as follows.
\begin{definition}
    For the game of $N$-users, if $r^{'}$ lies in the region where they need to share the caching cost, then the {\bf fair equilibrium} is a NE which satisfies: $x_n^m = \frac{S_m \hat{p}_n^m}{\sum_{k=1}^{N} \hat{p}_k^m}, \forall k \in \mathcal{N}$.
\end{definition}

Notice that this fair allocation is one of the equilibria. So, if users agree on this strategy before playing the game, none of them will have any incentive to deviate unilaterally. Figure \ref{fig:Fair_Allocation} (a) depicts the fair allocation solution for the example mentioned above. Notice that, for $0.4\leq r^{'}\leq0.6$, each user caches an amount proportional to his interest. The corresponding payments are shown in Figure \ref{fig:Fair_Allocation} (b). Comparing the results obtained from the fair allocation policy with the risk dominance results, we see that users payments are reduced. The fair allocation policy is one of the pre-play communication policies; however, we can find some other coordination approaches between users. For example, users can make a caching decision such that their corresponding payments are proportional to their interest. We summarize the optimal decentralized caching policy in Algorithm \ref{Alg:Optimal_Decentralized_Policy}.
\begin{figure}
    \centering
    \begin{tabular}{cc}
        \subfloat[Users caching decisions]{\includegraphics[width = 0.50\textwidth]{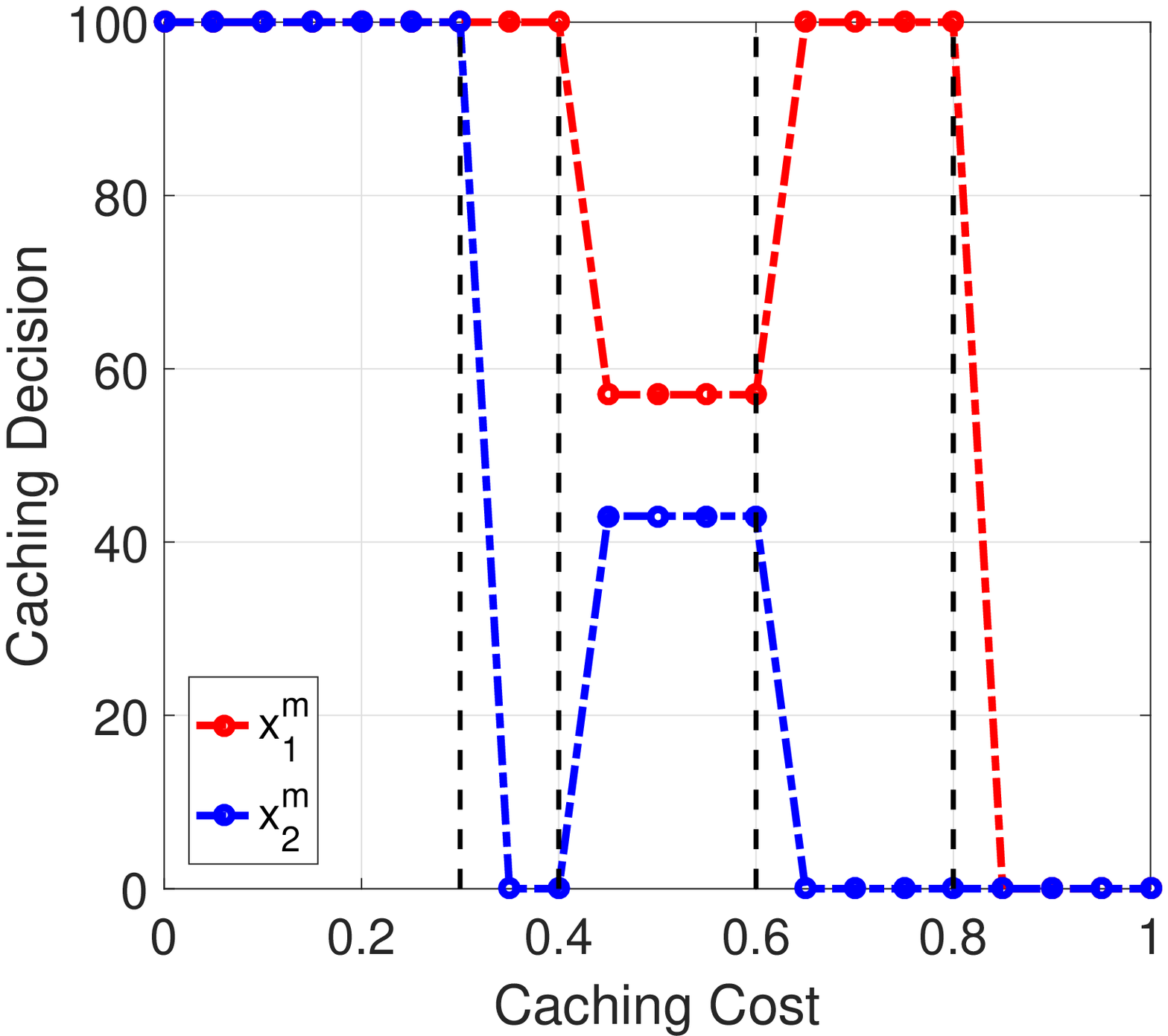}} &
        \subfloat[Expected user payment Comparison]{\includegraphics[width = 0.50\textwidth]{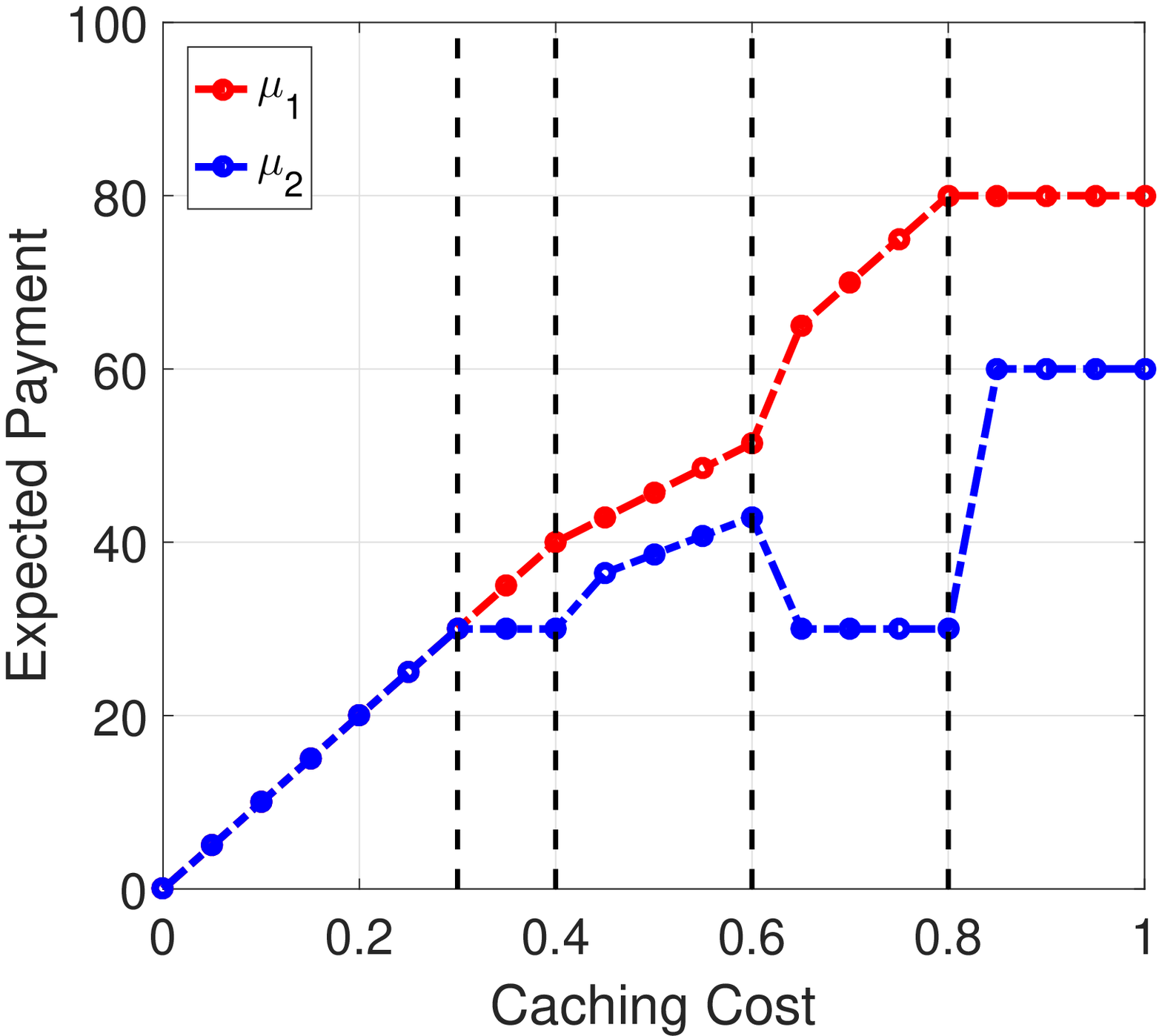}}
    \end{tabular}
    \caption{Fair Allocation example for $N=2,T=1$ in the decentralized caching scheme}
    \label{fig:Fair_Allocation}
\end{figure}

\begin{algorithm}[h!]
    \caption{Optimal Decentralized Caching Policy for $N$-users}
    \label{Alg:Optimal_Decentralized_Policy}
    \begin{algorithmic}
    \State \textbf{Given:} $N,M,L,\vec{\Pi}_n,\vec{\Theta}_{n},r^{'}$
    \For{$m=1$ to $M$}
        \Call{Caching}{$N,L,m,p_{n,t}^{m},\theta_{n,t}^{l}$}
    \EndFor
        \Procedure{Caching}{$N,L,m,p_{n,t}^{m},\theta_{n,t}^{l},r^{'}$}
        \For{$n=1$ to $N$}
            \If{$r^{'}\leq\frac{1}{T} \sum_{t=1}^{T} p_{n,t}^m \Bigl( 1 - v_{n,t} \Bigr)$} \State $x_n^m = S_m$
            \ElsIf{$r^{'}>\frac{1}{T} \sum_{t=1}^{T} p_{n,t}^m$} \State $x_n^m = S_m$
            \Else \State Set $Flag(n)=1$ for Partial Caching
            \EndIf
        \EndFor
        \EndProcedure
        \Procedure{PartialCaching}{$N,L,m,p_{n,t}^{m},\theta_{n,t}^{l},r^{'}$}
        \For{$n=1$ to $N$}
            \If {$Flag(n)=1$} \State $x_n^m=S_m \frac{\frac{1}{T}\sum_{t=1}^{T} p_{n,t}^m}{\sum_{k=1}^{N} Flag(k) \frac{1}{T}\sum_{t=1}^{T} p_{k,t}^m}$ \EndIf
        \EndFor
        \EndProcedure
    \end{algorithmic}
\end{algorithm}

\subsection{Choosing Optimal Memory Size}\label{Sec:Choosing_Memory}

In the previous section, we introduced the solution of the decentralized caching scheme based on the reward value $r$ (recall that $r^{'}=1-r$). Since, SP pays this reward back to all users, it will always try to reduce this amount as much as possible. But at the same time, this reward creates an incentive for users to participate in this model. We assume that each user has an isolated memory of size $Z_n$. For simplicity, we also assume that all users have the same memory size. Hence, the SP finds an aggregate memory of size $Z=NZ_n$. The SP has a reward preference to assign a certain reward $r$ corresponding to the assigned memory $Z$. We consider a linear relationship between $r$ and $Z$. In particular, we assume that $r(Z)=1-\gamma Z$, for some $\gamma > 0$. This means that the SP gives users more reward when they assign smaller memory and reduces the reward when they assign larger memory. This relation stops the users from increasing their memory and caching everything. At the same time, when the SP needs more memory, it can reduce the reward to push users towards increasing their memory size.

Now, considering this memory constraint, we can rewrite (\ref{Eq:User_Optimization_Problem}) as follows:
    \begin{equation}\label{Eq:User_Optimization_Problem_Modified}
        \begin{aligned}
            & \min \hspace{5mm} L_n^\mathcal{P}\\
            & \text{s.t.} \hspace{5mm} \sum_{m=1}^{M} x_n^m \leq Z_n.
        \end{aligned}
    \end{equation}
where $L_n^P$ is the peak load generated by user $n$. In particular, from (\ref{Eq:Payment_Proactive}), we see that $ L_n^\mathcal{P}=\mu_n^\mathcal{P} - r^{'} \sum_{m=1}^{M} x_n^m$. Converting this problem to an unconstrained problem, we get
    \begin{equation}\label{Eq:User_Optimization_Problem_Modified2}
        \begin{aligned}
            \underset{x_n^m}{\min} \hspace{3mm} L_n^\mathcal{P}+ r^{'} \left(\sum_{m=1}^{M} x_n^m - Z_n\right) =  \underset{x_n^m}{\min} \hspace{3mm} \left( L_n^\mathcal{P} + r^{'} \sum_{m=1}^{M} x_n^m \right)- r^{'}Z_n 
        \end{aligned}
    \end{equation}
where $r^{'}$ is the Lagrangian multiplier associated with the constraint  $\sum_{m=1}^{M} x_n^m \leq Z_n$. Notice that the first term in (\ref{Eq:User_Optimization_Problem_Modified2}) is the problem we solved in Section \ref{Sec:Optimal_Distributed} without having this memory constraint. Each user solves this optimization problem for all possible values of $Z_n$. The optimal choice of $Z_n$ depends on the SP reward preference. 

Plotting the Lagrangian multiplier $r^{'}$ for the solution of (\ref{Eq:User_Optimization_Problem_Modified2}) versus $Z_n$ we get the curve shown in Figure \ref{Fig:Choosing_Memory_UserN}. The optimal solution is determined by the intersection point between the SP reward preference and the users reward preference. We can see that $r^{'}$ takes the values of $\hat{p}_1^m$ or $\tilde{p}_1^m$. At the intersection point, and under the fair allocation scheme discussed in Section \ref{Sec:Fair_Distributed}, the optimal solution will be at $Z_n^*=\frac{S_m \hat{p}_1^m}{\sum_{n=1}^{N} \hat{p}_n^m}$. Considering all users, we will have the result shown in Figure \ref{Fig:Choosing_Memory}. The optimal solution $(Z^*,r^{'*})$ is at the intersection point between the SP reward preference and the users  reward preference. 

Note that each sub-region corresponds to one of the solutions shown in Figure \ref{fig:Opt_Soln_Char_N_T=1}. Therefore, the intersection point correspond to one of these solutions, where we may have some users are caching the content while others are sharing the cost of caching that content once between them. The best case scenario happens when the intersection leads to a solution where the content is cached once between all users. This means that each user caches a small portion of this data content, based on the relation between his interest and the aggregate interest of all users. This also yields a lower memory consumption as the number of users increases.

\begin{figure}[h!]
	\centering
      \begin{tikzpicture}[scale=0.75]
            
            \draw[->] (0,0) -- (11,0) node[anchor=north] {$Z_n$};
            \draw	(0,0) node[anchor=north] {$0$}
				(6,0) node[anchor=north] {\textbf{$Z_n^*=\frac{S_m\hat{p}_1^m}{\sum_{n=1}^{N}\hat{p}_n^m}$}}
                	(10,0) node[anchor=north] {$S_m$};
            \draw[->] (0,0) -- (0,11) node[anchor=east] {$r$};
            \draw	(0,0) node[anchor=east] {$0$}
				(0,3) node[anchor=east] {$\tilde{p}_1^m$}
				(0,6) node[anchor=east] {\textbf{$r^*$}}
              		(0,10) node[anchor=east] {$\hat{p}_1^m$};           
                
            \draw[step=1cm,gray,very thin] (0,0) grid (10,10);            

            \draw[line width=0.5mm,red] (0,10) -- (6,10);
            \draw[line width=0.5mm,red] (6,10) -- (6,3);
            \draw[line width=0.5mm,red] (6,3) -- (10,3);
            \draw[line width=0.5mm,red] (10,3) -- (10,0);
            \draw[line width=0.5mm,blue] (0,0) -- (10,10); 

            \draw [fill=red] (6,6) circle (0.1cm);

            \draw (3,9.5) node[red] {\scriptsize{User $(n)$ Reward Decision}};
            \draw  (3,3) node[blue,rotate=45,anchor=south] {\scriptsize{SP Reward Preference $r(Z_n)$}};
              
        \end{tikzpicture}	
	\caption{Choosing an optimal memory ($Z_n$) for user $(n)$ in the decentralized caching scheme}
	\label{Fig:Choosing_Memory_UserN}
\end{figure}
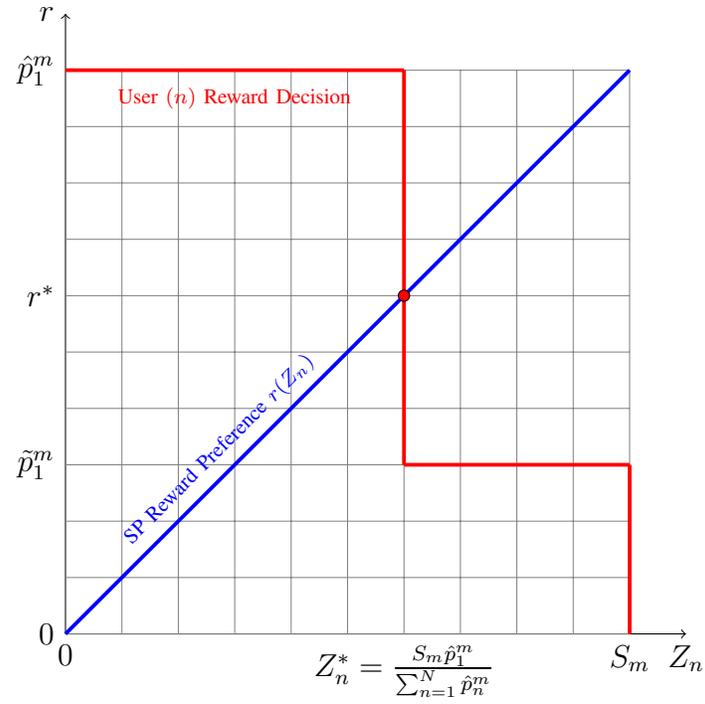

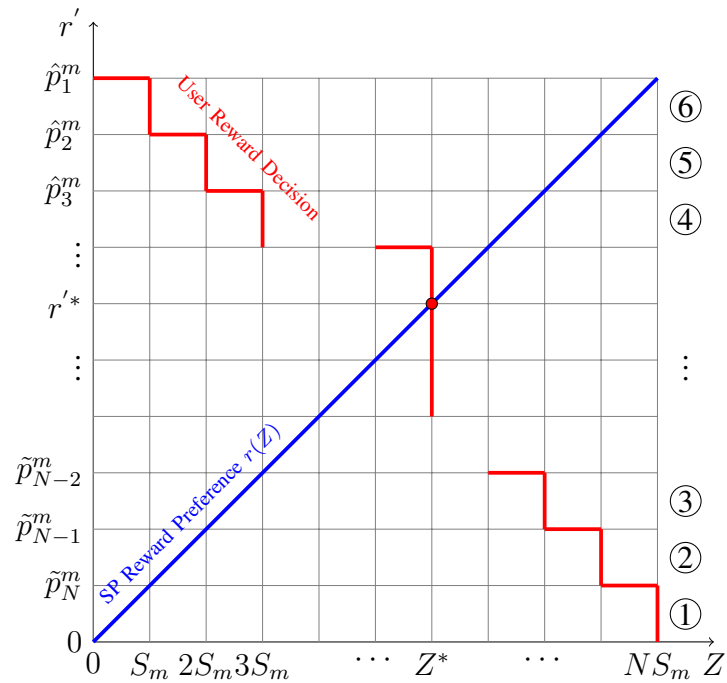
\begin{figure}[h!]
	\centering
      \begin{tikzpicture}[scale=0.75]
            
            \draw[->] (0,0) -- (11,0) node[anchor=north] {$Z$};
            \draw	(0,0) node[anchor=north] {$0$}
				(1,0) node[anchor=north] {$S_m$}
				(2,0) node[anchor=north] {$2S_m$}
				(3,0) node[anchor=north] {$3S_m$}
            		(5,0) node[anchor=north] {$\cdots$}
				(6,0) node[anchor=north] {\textbf{$Z^*$}}
            		(8,0) node[anchor=north] {$\cdots$}
                	(10,0) node[anchor=north] {$NS_m$};
            \draw[->] (0,0) -- (0,11) node[anchor=east] {$r^{'}$};
            \draw	(0,0) node[anchor=east] {$0$}
				(0,1) node[anchor=east] {$\tilde{p}_N^m$}
				(0,2) node[anchor=east] {$\tilde{p}_{N-1}^m$}
				(0,3) node[anchor=east] {$\tilde{p}_{N-2}^m$}
            	(0,5) node[anchor=east] {$\vdots$}
				(0,6) node[anchor=east] {\textbf{$r^{'*}$}}
            	(0,7) node[anchor=east] {$\vdots$}
				(0,8) node[anchor=east] {$\hat{p}_3^m$}
				(0,9) node[anchor=east] {$\hat{p}_2^m$}
                (0,10) node[anchor=east] {$\hat{p}_1^m$};           
                
            \draw[step=1cm,gray,very thin] (0,0) grid (10,10);            

            \draw[line width=0.5mm,red] (0,10) -- (1,10);
            \draw[line width=0.5mm,red] (1,10) -- (1,9);
            \draw[line width=0.5mm,red] (1,9) -- (2,9);
            \draw[line width=0.5mm,red] (2,9) -- (2,8);
            \draw[line width=0.5mm,red] (2,8) -- (3,8);
            \draw[line width=0.5mm,red] (3,8) -- (3,7);
            \draw[line width=0.5mm,red] (5,7) -- (6,7);
            \draw[line width=0.5mm,red] (6,7) -- (6,4);
            \draw[line width=0.5mm,red] (7,3) -- (8,3);
            \draw[line width=0.5mm,red] (8,3) -- (8,2);
            \draw[line width=0.5mm,red] (8,2) -- (9,2);
            \draw[line width=0.5mm,red] (9,2) -- (9,1);
            \draw[line width=0.5mm,red] (9,1) -- (10,1);
            \draw[line width=0.5mm,red] (10,1) -- (10,0);

            \draw[line width=0.5mm,blue] (0,0) -- (10,10); 

            \draw [fill=red] (6,6) circle (0.1cm);

            \draw (10.5,0.5) node[black] {$\circled{1}$};
            \draw (10.5,1.5) node[black] {$\circled{2}$};
            \draw (10.5,2.5) node[black] {$\circled{3}$};
            \draw (10.5,5) node[black] {$\vdots$};
            \draw (10.5,7.5) node[black] {$\circled{4}$};
            \draw (10.5,8.5) node[black] {$\circled{5}$};
            \draw (10.5,9.5) node[black] {$\circled{6}$};
            \draw (3,9) node[red,rotate=-45,anchor=north] {\scriptsize{User Reward Decision}};
            \draw  (2,2) node[blue,rotate=45,anchor=south] {\scriptsize{SP Reward Preference $r(Z)$}};
              
        \end{tikzpicture}	
	\caption{Choosing an optimal memory ($Z$) for $N$ users in the decentralized caching scheme}
	\label{Fig:Choosing_Memory}
\end{figure}

\subsection{Impact of User Mobility}

Users mobility statistics affect the optimal solution of the decentralized caching policy. The optimal decision of each user depends on his meeting probabilities with other users. The user who is meeting others with a higher probability will have more potential for partial caching. In particular, when his meeting probabilities increase, the value of $\tilde{p}_n^m$ decreases and the region of partial caching increases. To see this, let us consider a special case when users have similar mobility patterns. Further, we consider the case when users visit all locations with the same probability and hence have the same meeting probability. In particular, consider the case when $\theta_{1,t}^{l}=\theta_{2,t}^{l}=\cdots=\theta_{N,t}^{l}=\theta_{t}^{l}=\frac{1}{L}, \forall t \in \{1,2,\cdots,T\}$. Therefore, we have 
\begin{equation*}
    v_{n,t} = \sum_{k=1}^{N-1} (-1)^{k+1} \binom{N-1}{k} \sum_{l=1}^{L} \frac{1}{L^{k+1}}, \forall n,t
\end{equation*}

When $L\rightarrow\infty$, we have $v_{n,t}\rightarrow0,\forall n,t$. This case is typically similar to the proactive caching model discussed in \cite{Hosny2015Towards}. Since we are assuming here that $\alpha_m=1,\forall m$, each user will cache the content when his interest exceeds the caching cost, regardless of the other users decision. Note that $r^{'}$ is similar to $\frac{y_o}{y_p}$ in the proactive caching model, since we assume that $y_p=1$. For example, the solution of $N=2$ will be as shown in Figure \ref{Fig:Opt_Soln_N=2_No_Mobility}. When $r^{'}$ is smaller than $\hat{p}_1^m,\hat{p}_2^m$ both users cache this content. When it exceeds $\hat{p}_1^m$ or $\hat{p}_2^m$, the corresponding user opts to avoid caching.
    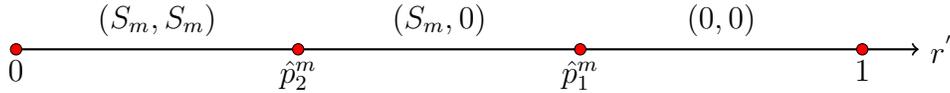
\begin{figure}[h!]
        \centering
        \begin{tikzpicture}[scale=0.75]
            \draw[->,thick] (0,0) -- (16,0) node[anchor=west] {$r^{'}$};
            \draw	(0,0) node[anchor=north] {$0$}
        		    (5,0) node[anchor=north] {$\hat{p}_2^m$}
        		    (10,0) node[anchor=north] {$\hat{p}_1^m$}
        		    (15,0) node[anchor=north] {$1$};
            \draw	(2.5,1) node[anchor=north] {$(S_m,S_m)$}
        		    (7.5,1.) node[anchor=north] {$(S_m,0)$}
        		    (12.5,1) node[anchor=north] {$(0,0)$};
            \draw [fill=red] (0,0) circle (0.1cm);
            \draw [fill=red] (5,0) circle (0.1cm);
            \draw [fill=red] (10,0) circle (0.1cm);
            \draw [fill=red] (15,0) circle (0.1cm);
            
        \end{tikzpicture}
	    \caption{Optimal decentralized caching policy for $N=2$ as $L\rightarrow\infty$}
        \label{Fig:Opt_Soln_N=2_No_Mobility}
    \end{figure}
    
Now suppose that users are moving together such that their meeting probabilities are very close to 1. This is similar to the content trading model. The difference is that all users are setting their selling price to 0, i.e. they are sharing their proactive downloads for free. For example, the optimal solution for $N=2$ shown in Figure \ref{Fig:Opt_Soln_Char_N=2_T=1_Case_2} will be modified as shown in Figure \ref{Fig:Opt_Soln_N=2_Connected}. Since, users are sharing their proactive downloads for free, partial caching will be an optimal solution, instead of having one user caching the content and selling it to all other users.     \begin{figure}[h!]
        \centering
        \begin{tikzpicture}[scale=0.75]
            \draw[->,thick] (0,0) -- (16,0) node[anchor=west] {$r^{'}$};
            \draw	(0,0) node[anchor=north] {$0$}
        		    (5,0) node[anchor=north] {$\hat{p}_2^m$}
        		    (10,0) node[anchor=north] {$\hat{p}_1^m$}
        		    (15,0) node[anchor=north] {$1$};
            \draw	(2.5,1) node[anchor=north] {$(S_m-x,x)$}
        		    (7.5,1.) node[anchor=north] {$(S_m,0)$}
        		    (12.5,1) node[anchor=north] {$(0,0)$};
            \draw [fill=red] (0,0) circle (0.1cm);
            \draw [fill=red] (5,0) circle (0.1cm);
            \draw [fill=red] (10,0) circle (0.1cm);
            \draw [fill=red] (15,0) circle (0.1cm);
            
        \end{tikzpicture}
	    \caption{Optimal decentralized caching policy for $N=2$ as $L\rightarrow1$}
        \label{Fig:Opt_Soln_N=2_Connected}
    \end{figure}





\section{Conclusion} \label{Sec:Conclusion}

We considered a mobility-aware D2D caching network where caching decision is taken based on the users demand and mobility statistics. Two caching schemes, centralized and decentralized, were considered. We started by considering a \emph{centralized D2D caching network}, where the SP is pushing data items in users devices and pays them a reward for participation. The SP aim was to minimize its incurred service cost by harnessing user's demand and mobility statistics. An \emph{optimal caching policy} was introduced that allows the SP to enhance its caching decisions. The \emph{complexity} of the optimal caching policy was found to grow exponentially with the number of users. Therefore, we introduced a \emph{greedy caching policy} that has a polynomial order complexity. The proposed greedy algorithm was used to establish \emph{upper and lower bounds} on the gain achieved by the optimal caching policy. The optimal solution of the proposed model was found to depend on \emph{users reward preference} which affects the assigned memory in their devices.
Our vision was completed by considering a \emph{decentralized D2D caching network}, where users make the caching decision based on the SP reward. We introduced an \emph{optimal caching policy} that allows users to minimize their expected payment. We formulated the tension between the SP and users as a \emph{Stackelberg game}. Best response analysis was used to identify a \emph{subgame perfect Nash equilibrium} between users. The optimal solution of the proposed model was found to depend on the \emph{SP reward preference}, which affects the assigned memory in users devices. We found some regimes for the reward value where the SPNE was non-unique. A \emph{fair allocation} caching policy was adopted to choose one of these SPNEs. 

To understand the impact of user behavior, we considered some special cases when users have similar behavior. If users have identical behavior, i.e. they have similar interest and mobility statistics, they receive a similar amount of caching. Moreover, when the number of popular locations grows large, their meeting probability approaches zero and the amount of caching depends on their interest only. We showed that if users have similar interest and different mobility statistics, caching the same content at all of them happens only if their meeting probability is zero. In this case, the amount of data cached at each user depends on their interest only. Users who are meeting each other with a probability of $1$, split the content caching between them and caching the data item once is optimal. Fair allocation plays an important role here, to pick one of the SPNEs. If users have similar and uniform mobility patterns, i.e. they visit all popular locations with the same probability, the complexity of the centralized optimal caching policy was significantly reduced. We used this special case to show how the mobility-aware model simplifies the proactive caching and the content trading models. Our objective from this part was to explore how users mobility statistics affect the caching decision. 

The results of this work extend our understanding for users behavior in D2D caching networks and allow us to add mobility dynamics to our content trading model discussed in \cite{Hosny2015Towards}. However, there are many aspects that need more investigations. Capturing group mobility in D2D caching networks helps SPs to leverage more statistics about users. This problem has its own importance in studying the correlation between users and how to exploit it to enhance the network performance. We considered mainly the economics point of view in the previous work and discussed cost minimization in the individual mobility model. There are some other metrics that can be used to evaluate the proposed models from different angles, like outage probability and achievable throughput. Further, scaling behavior of such networks is a major point in this direction. We need to investigate the performance of the network when it expands to a larger number of users or data items.

We also need to study the cooperative and distributed caching in social-aware D2D caching networks which is another dimension to capture the correlation between users. Inspired by the main results and insights from the group mobility direction, we can extend it to grasp another parameter which affects the D2D caching networks. The carrier should be able to harness the statistics about relations between users to optimize the cached data items. This should be another thrust towards cost minimization. The SP can also exploit this aspect to shape users demand and consequently maximizes its profit. Over and above, users gain from their relationships with others in many ways. By reducing the service cost, the SP will have more potential to offer lower prices to users, as a way to shape their demand. There will be a higher possibility to find the requested data items among users in the same vicinity.



\begin{thebibliography}{1}
	\providecommand{\url}[1]{#1}
	\csname url@samestyle\endcsname
	\providecommand{\newblock}{\relax}
	\providecommand{\bibinfo}[2]{#2}
	\providecommand{\BIBentrySTDinterwordspacing}{\spaceskip=0pt\relax}
	\providecommand{\BIBentryALTinterwordstretchfactor}{4}
	\providecommand{\BIBentryALTinterwordspacing}{\spaceskip=\fontdimen2\font plus
		\BIBentryALTinterwordstretchfactor\fontdimen3\font minus
		\fontdimen4\font\relax}
	\providecommand{\BIBforeignlanguage}[2]{{%
			\expandafter\ifx\csname l@#1\endcsname\relax
			\typeout{** WARNING: IEEEtran.bst: No hyphenation pattern has been}%
			\typeout{** loaded for the language `#1'. Using the pattern for}%
			\typeout{** the default language instead.}%
			\else
			\language=\csname l@#1\endcsname
			\fi
			#2}}
	\providecommand{\BIBdecl}{\relax}
	\BIBdecl
	
	\bibitem{Hosny2015Game}
	F.~Alotaibi, S.~Hosny, H.~E. Gamal, and A.~Eryilmaz, ``A game theoretic
	approach to content trading in proactive wireless networks,'' in \emph{2015
		IEEE International Symposium on Information Theory (ISIT)}, June 2015, pp.
	2216--2220.
	
	\bibitem{cisco_2015}
	C.~V.~N. Index, ``Cisco visual networking index: Global mobile data traffic
	forecast update, 2015--2020 white paper,'' Tech. rep. Cisco, 2016. url:
	http://www. cisco.
	com/c/en/us/solutions/collateral/service-provider/visual-networking-index-vni/mobile-white-paper-c11-520862.
	html (visited on 03/26/2016)(cit. on p. 6), Tech. Rep., 2016.
	
	\bibitem{Federal2002Spectrum}
	FCC, ``Spectrum policy task force report, fcc 02-155,'' 2002.
	
	\bibitem{Song2010Limits}
	\BIBentryALTinterwordspacing
	C.~Song, Z.~Qu, N.~Blumm, and A.-L. Barabási, ``Limits of predictability in
	human mobility,'' \emph{Science}, vol. 327, no. 5968, pp. 1018--1021, 2010.
	[Online]. Available:
	\url{http://www.sciencemag.org/content/327/5968/1018.abstract}
	\BIBentrySTDinterwordspacing
	
	\bibitem{farrahi2008discovering}
	K.~Farrahi and D.~Gatica-Perez, ``Discovering human routines from cell phone
	data with topic models,'' in \emph{Wearable Computers, 2008. ISWC 2008. 12th
		IEEE International Symposium on}.\hskip 1em plus 0.5em minus 0.4em\relax
	IEEE, 2008, pp. 29--32.
	
	\bibitem{Tadrous2013Proactive}
	J.~Tadrous, A.~Eryilmaz, and H.~El~Gamal, ``Proactive resource allocation:
	Harnessing the diversity and multicast gains,'' \emph{Information Theory,
		IEEE Transactions on}, vol.~59, no.~8, pp. 4833--4854, Aug 2013.
	
	\bibitem{yu2011resource}
	C.-H. Yu, K.~Doppler, C.~B. Ribeiro, and O.~Tirkkonen, ``Resource sharing
	optimization for device-to-device communication underlaying cellular
	networks,'' \emph{Wireless Communications, IEEE Transactions on}, vol.~10,
	no.~8, pp. 2752--2763, 2011.
	
	\bibitem{Caire2016Review}
	G.~C. M.~Ji and A.~F. Molisch, ``Wireless device-to-device caching networks:
	Basic principles and system performance,'' \emph{IEEE Journal on Selected
		Areas in Communications}, vol.~34, no.~1, pp. 176--189, Jan 2016.
	
	\bibitem{Caire2016CachingD2D}
	M.~Ji, G.~Caire, and A.~F. Molisch, ``Fundamental limits of caching in wireless
	d2d networks,'' \emph{IEEE Transactions on Information Theory}, vol.~62,
	no.~2, pp. 849--869, Feb 2016.
	
	\bibitem{Maddah2014Fundamental}
	M.~A. Maddah-Ali and U.~Niesen, ``Fundamental limits of caching,'' \emph{IEEE
		Transactions on Information Theory}, vol.~60, no.~5, pp. 2856--2867, May
	2014.
	
	\bibitem{Caire2013Optimal}
	M.~Ji, G.~Caire, and A.~F. Molisch, ``Optimal throughput-outage trade-off in
	wireless one-hop caching networks,'' in \emph{Information Theory Proceedings
		(ISIT), 2013 IEEE International Symposium on}, July 2013, pp. 1461--1465.
	
	\bibitem{Caire2015Throughput}
	------, ``The throughput-outage tradeoff of wireless one-hop caching
	networks,'' \emph{IEEE Transactions on Information Theory}, vol.~61, no.~12,
	pp. 6833--6859, Dec 2015.
	
	\bibitem{DavidTse2002Mobility}
	M.~Grossglauser and D.~N.~C. Tse, ``Mobility increases the capacity of ad hoc
	wireless networks,'' \emph{IEEE/ACM Transactions on Networking}, vol.~10,
	no.~4, pp. 477--486, Aug 2002.
	
	\bibitem{Gupta2000Capacity}
	P.~Gupta and P.~R. Kumar, ``The capacity of wireless networks,'' \emph{IEEE
		Transactions on Information Theory}, vol.~46, no.~2, pp. 388--404, Mar 2000.
	
	\bibitem{Camp02asurvey}
	T.~Camp, J.~Boleng, and V.~Davies, ``A survey of mobility models for ad hoc
	network research,'' \emph{Wireless Communications and Mobile Computing
		(WCMC)}, vol.~2, no.~5, pp. 483--502, Aug 2002.
	
	\bibitem{Hosny2015Towards}
	F.~Alotaibi, S.~Hosny, J.~Tadrous, H.~E. Gamal, and A.~Eryilmaz, ``Towards a
	marketplace for mobile content: Dynamic pricing and proactive caching,''
	\emph{arXiv preprint arXiv:1511.07573}, 2015.
	
	\bibitem{Chiang98Mobility}
	C.-C. Chiang, ``Wireless network multicasting,'' Ph.D. dissertation, University
	of California, Los Angeles, 1998.
	
	\bibitem{harsanyi1988general}
	J.~C. Harsanyi, R.~Selten \emph{et~al.}, ``A general theory of equilibrium
	selection in games,'' \emph{MIT Press Books}, vol.~1, 1988.
	
	\bibitem{fevrier2006equilibrium}
	P.~F{\'e}vrier and L.~Linnemer, ``Equilibrium selection: payoff or risk
	dominance?: the case of the “weakest link”,'' \emph{Journal of Economic
		Behavior \& Organization}, vol.~60, no.~2, pp. 164--181, 2006.
	
\end{thebibliography}

\end{document}